\documentclass{lmcs}
\pdfoutput=1

\usepackage{lastpage}
\lmcsdoi{21}{4}{25}
\lmcsheading{}{\pageref{LastPage}}{}{}%
{Apr.~08,~2024}{Dec.~01,~2025}{}

\usepackage{amsmath}
\usepackage{mathtools}
\usepackage{amssymb}
\usepackage{courier}
\usepackage[T1]{fontenc}
\usepackage{graphicx}
\usepackage{hyperref} 
\usepackage[utf8]{inputenc}
\usepackage{listings}
\usepackage{mathpartir}
\usepackage{stmaryrd}
\usepackage{subcaption}
\usepackage{tikz-cd}
\usepackage{tikz}
\usetikzlibrary{cd, fit}
\usepackage{xspace}
\usepackage{booktabs}

\newcommand{\set}[1]{\{#1\}}

\newcommand{\kord}[1]{\mathtt{#1}}
\newcommand{\kop}[1]{\;\mathtt{#1}\;}
\newcommand{\kpre}[1]{\mathtt{#1}\;}

\newcommand{\type}[1]{\mathtt{#1}}
\newcommand{\hto}{\Rightarrow}
\newcommand{\dirt}{\Delta}

\newcommand{\langName}[1]{\textsc{#1}\xspace}
\newcommand{\eff}{\langName{Eff}}
\newcommand{\ocaml}{\langName{OCaml}}
\newcommand{\coreEff}{\langName{CoreEff}}

\newcommand{\tyCtx}{\Gamma}
\newcommand{\ctx}{\Xi}
\newcommand{\ctxuse}{\ctx_{\mathrm{use}}}
\newcommand{\coercion}{\gamma}
\newcommand{\coercionN}[1]{\coercion_{#1}}
\newcommand{\coerVal}{\coercionN{\mathtt{v}}}
\newcommand{\coerComp}{\coercionN{\mathtt{c}}}

\newcommand{\coerDirt}{\coercionN{\mathtt{d}}}

\newcommand{\coerSubst}{Y}
\newcommand{\tyvar}{\alpha}
\newcommand{\dirtvar}{\delta}
\newcommand{\tycoervar}{\omega}
\newcommand{\dirtcoervar}{\varpi}

\newcommand{\op}[1][op]{{#1}}
\newcommand{\ops}{\mathcal{O}}

\newcommand{\cast}[2]{{#1} \vartriangleright {#2}}
\newcommand{\operation}[4][op]{\op[#1]~{#2}~({#3}. {#4})}

\newcommand{\fun}[1]{\kpre{fun} #1 \mapsto}
\newcommand{\keyreturn}{\kord{return}}

\newcommand{\keydo}{\kord{do}}

\newcommand{\return}[1]{\keyreturn~{#1}}                              
\newcommand{\doin}[3]{\keydo~{#1} \leftarrow {#2} ;{#3}}

\newcommand{\cupDirtBoth}{\mathrel{\cup}} 
\newcommand{\cupDirtSingle}{\mathrel{\cup\!^+}} 
\renewcommand{\setminus}{\smallsetminus}

\newcommand{\tuple}[1]{\left(#1\right)}

\newcommand{\skel}{\mathit{S}}
\newcommand{\skelvar}{\varsigma}

\newcommand{\refl}[1]{\langle {#1} \rangle}
\newcommand{\trgUnitRefl}{\refl{\tyUnit}}

\newcommand{\tyUnit}{\type{Unit}}
\newcommand{\tyBool}{\type{Bool}}
\newcommand{\tyInt}{\type{Int}}
\newcommand{\tyFloat}{\type{Float}}
\newcommand{\dirtyTo}[1]{\stackrel{#1}{\to}}
\newcommand{\longDirtyTo}[1]{\xrightarrow{#1}}
\newcommand{\tmUnit}{\kord{unit}}

\newcommand{\vty}{\vtya}                   
\newcommand{\vtya}{\mathit{A}}             
\newcommand{\cty}{\underline{\mathit{C}}}  
\newcommand{\dirty}[2]{#1~!~#2}

\newcommand{\vdashNamedD}[1]{\vdash_{{\scriptscriptstyle \hspace{-1mm}\mathtt{#1}}}}

\newcommand{\tCoerty}[3]{{#1} \vdash  {#2}: {#3}}
\newcommand{\tCoerVal}[3]{{#1} \vdash {#2} : {#3}}
\newcommand{\tCoerComp}[3]{{#1} \vdash {#2} : {#3}}
\newcommand{\tCoerDirt}[3]{{#1} \vdash {#2} : {#3}}
\newcommand{\tCoerSub}[3]{{#1} \vdash {#2} : {#3}}

\newcommand{\tVty}[3]{{#1} \vdash  {#2} : {#3}}
\newcommand{\tCty}[3]{{#1} \vdash  {#2} : {#3}}
\newcommand{\tSty}[2]{{#1} \vdash  {#2} \mathtt{~skel}}
\newcommand{\tDty}[2]{{#1} \vdash  {#2} \mathtt{~dirt}}

\newcommand{\tTyCtx}[2]{{#1} \vdash {#2} \mathtt{~tyCtx}}
\newcommand{\tCtx}[1]{\vdash {#1} \mathtt{~ctx}}
\newcommand{\tCtxn}[3][]{#1 \vdash_{{#2}} {#3} \mathtt{~ctx}}

\newcommand{\mtVal}[3]{{#1} \vdash {#2} : {#3}}
\newcommand{\mtComp}[3]{{#1} \vdash {#2} : {#3}}
\newcommand{\tVal}[4]{\mtVal{{#1}; {#2}}{#3}{#4}}
\newcommand{\tComp}[4]{\mtComp{{#1}; {#2}}{#3}{#4}}

\newcommand{\intp}{\eta}

\newcommand{\sem}[2][\xi]{\llbracket{#2} \rrbracket_{#1}}
\newcommand{\ssem}[2][\xi]{\llparenthesis{#2} \rrparenthesis_{#1}}
\newcommand{\env}{\vec{a}}

\newcommand{\monadsig}{\mathcal{S}}

\newcommand{\inject}[2][\xi]{\iota_{\sem[#1]{#2}}}

\newcommand{\inval}{\text{in}_{\keyreturn}}
\newcommand{\inop}[1][op]{\text{in}_{\mathit{#1}}}

\newcommand{\fp}[1]{\mathtt{fp}(#1)}

\newcommand{\FP}{F}

\newcommand{\validSub}[3]{\vdashNamedD{} #1 : #2 \Rightarrow #3}
\newcommand{\emptysub}{\emptyset}
\newcommand{\emptyctx}{\varepsilon}

\newcommand{\prog}[1]{\texttt{#1}}

\newcommand{\phase}{\Phi}
\newcommand{\subPhase}{\varphi}

\lstdefinestyle{eff}{%
  language=Caml,
  moredelim=*[is][\itshape]{/@}{@/},
  numbers=none,mathescape=true,showstringspaces=false,
  morekeywords={handle,handler,with,return,perform},
  keywordstyle=\bfseries,
xleftmargin=1em,basicstyle=\ttfamily\small}

\lstdefinestyle{ocaml}{%
  language=Caml,
  moredelim=*[is][\itshape]{/@}{@/},
  numbers=none,mathescape=true,showstringspaces=false,
  keywordstyle=\bfseries,
xleftmargin=1em,basicstyle=\ttfamily\small}

\lstnewenvironment{source}{\lstset{
    basicstyle=\ttfamily\small,
}} {}
\newcommand{\effcode}[1]{\textcolor{purple}{\lstinline{#1}}}
\newcommand{\ocamlcode}[1]{\textcolor{teal}{\lstinline[style=ocaml]{#1}}}
\lstnewenvironment{efflisting}{\lstset{style=eff,basicstyle=\linespread{0.999}\ttfamily\small\color{purple}}}{}
\lstnewenvironment{ocamllisting}{\lstset{style=ocaml,basicstyle=\linespread{0.999}\ttfamily\small\color{teal}}}{}

\newcommand{\hilite}[1]{\tikz[baseline=(X.base)] \node[rectangle, fill=gray!50, rounded corners, inner sep=1mm] (X) {\ensuremath{#1}};}
\newcommand\hilitecast[2]{\cast{#1}{\hilite{\vphantom{!}#2}}}

\graphicspath{{imgs/}}

\newcommand{\todo}[2][]{\relax}

\begin{document}

\title[Simplifying explicit subtyping coercions]{Simplifying explicit subtyping coercions \texorpdfstring{\\}{} in a polymorphic calculus with effects}
\thanks{This material is based upon work supported by the Air Force Office of Scientific Research under awards number FA9550-17-1-0326 and FA9550-21-1-0024.}

\author[F.~Koprivec]{Filip Koprivec\lmcsorcid{0000-0001-9624-9369}}[a,b]
\author[M.~Pretnar]{Matija Pretnar\lmcsorcid{0000-0001-7755-2303}}[a,b]

\address{University of Ljubljana, Faculty of Mathematics and Physics, Jadranska 19, SI-1000 Ljubljana, Slovenia}
\address{Institute of Mathematics, Physics and Mechanics, Jadranska 19, SI-1000 Ljubljana, Slovenia}

\email{filip.koprivec@fmf.uni-lj.si}
\email{matija.pretnar@fmf.uni-lj.si}

\begin{abstract}
  \noindent
  Algebraic effect handlers are becoming an increasingly popular way of structuring effectful computations, and their performance is often a concern.
  One of the proposed approaches towards efficient compilation is tracking effect information through explicit subtyping coercions.
  However, in the presence of polymorphism, these coercions are compiled into additional arguments of compiled functions, incurring significant overhead.

  In this paper, we present a polymorphic effectful calculus, identify simplification phases needed to reduce the number of unnecessary constraints, and prove that they preserve semantics. In addition, we implement the simplification algorithm in the \eff language and evaluate its performance on a number of benchmarks. Though we do not prove the optimality of the presented simplifications, the results show that the algorithm eliminates all coercions, resulting in code as efficient as manually monomorphised one.
\end{abstract}



\maketitle

\section*{Introduction}
\label{sec:introduction}
Recent years have seen an increase in the number of programming languages that support algebraic effect handlers~\cite{DBLP:journals/acs/PlotkinP03,DBLP:journals/corr/PlotkinP13}.
With a widespread usage, the need for performance is becoming ever more important.
And there are two main ways of achieving it:
an efficient low-level runtime~\cite{DBLP:conf/pldi/Sivaramakrishnan21,DBLP:journals/pacmpl/MaGLZ24,10.1145/3747529,10.1145/3763177}
or an optimising compiler to a high-level language~\cite{DBLP:journals/pacmpl/SchusterBO20,DBLP:journals/pacmpl/XieL21,DBLP:journals/pacmpl/KarachaliasKPS21}, which we focus on in this paper.

Our recent work~\cite{DBLP:journals/pacmpl/KarachaliasKPS21} has shown how an optimising compiler can take code written using the full flexibility of handlers, infer precise information about which parts of it use effects and which are pure, and produce code that matches conventional handcrafted one.
The approach tracks effect information through explicit subtyping coercions~\cite{DBLP:journals/jfp/KarachaliasPSVS20}, inserted around almost every sub-term, similar to where type constraints are introduced in Hindley-Milner type inference.
In monomorphic code, most of these coercions become trivial and disappear, but for polymorphic functions, they need to be passed around as additional parameters.
As a result, the number of such parameters grows linearly with the size of the function body, rendering performance practically unusable.

In this paper, we propose an algorithm that soundly reduces (and often completely eliminates) redundant coercion parameters, leading to a performance comparable to that of monomorphic code.
We start with an overview of the approach (Section~\ref{sec:overview}) and continue with a specification of our working language (Section~\ref{sec:language}).
Afterwards, we turn to our contributions, which are:
\begin{itemize}
  \item Identifying requirements for a simplifying substitution to be correct with respect to typing (Sections~\ref{sec:substitutions} and~\ref{sec:polarity}).
  \item A number of progressive phases for simplification of constraints (Section~\ref{sec:simplifications}).
  \item A proof that all the presented simplifications preserve the denotational semantics (Section~\ref{sec:denotationalSemantics}).
  \item An implementation of the algorithm in a prototype language \eff and an evaluation of the impact it has on the code size and runtime (Section~\ref{sec:implementation}).
\end{itemize}
We conclude by discussing related and future work (Section~\ref{sec:conclusion}).
We assume readers to be familiar with algebraic effects and handlers, and refer them to a tutorial~\cite{DBLP:journals/entcs/Pretnar15} in case they are not.
Unless explicitly noted, all proofs proceed by a straightforward structural induction.
More details can be found in the first author's PhD~thesis~\cite{Koprivec_2024}.

\section{Overview}
\label{sec:overview}
\subsection{Explicit effect subtyping}

To see an example of the compilation pipeline in action, assume an operation \effcode{Random : unit -> bool} that produces random boolean values each time when triggered, and \effcode{Print : string -> unit} that prints a given string.
Then, a function that randomly prints a given string is written in \eff as:
\begin{efflisting}
let print_randomly s =
  if (perform (Random ())) then (perform (Print s)) else ()
\end{efflisting}
where the \effcode{perform} keyword triggers one of the operations.

The above function is first translated to \eff's core language \coreEff, which is a fine-grain call-by-value language~\cite{DBLP:journals/iandc/LevyPT03}, meaning it distinguishes values and computations.
For example, since \effcode{perform (Random ())} is a computation, and a conditional statement expects a boolean value, we must explicitly sequence the two computations and get:
\begin{align*}
  & \fun{(s : \type{String})}                \\
  & \quad \doin{b}{\op[Random] \, \tmUnit}{} \\
  & \quad \kpre{if} b \kop{then}             \\
  & \quad\quad \op[Print] \, s               \\
  & \quad \kop{else}                         \\
  & \quad\quad \return{\tmUnit}
\end{align*}
As effects play a key role in the language, \coreEff computations are assigned types of the form $\dirty{\vty}{\dirt}$, where $\vty$ is the type of returned values, while a \emph{dirt} $\dirt = \set{\op_1, \dots, \op_n }$ captures the set of operations that the computation may perform.

In the above example, $\op[Random] \, \tmUnit$ is assigned the type $\dirty{\tyBool}{\set{\op[Random]}}$, while $\return{\tmUnit}$ is assigned the type $\dirty{\tyUnit}{\emptyset}$, and the fully annotated type of the function is $\type{String} \to \dirty{\tyUnit}{\set{\op[Random], \op[Print]}}$, which we abbreviate as $\type{String} \longDirtyTo{\set{\op[Random], \op[Print]}} \tyUnit$.

To accommodate for the differences in assigned effects, e.g. in the two conditional branches above, \eff features subtyping, which is a more flexible approach than row-based effect systems, as the latter can suffer from the \emph{poisoning problem}~\cite{DBLP:conf/popl/WansbroughJ99}, as unification propagates effectful annotations even into pure parts of the program.

When inferring types, the inference algorithm also annotates certain terms with type coercions, which witness the compatibility of all types and effects.
A fully annotated example with coercions highlighted in \hilite{gray} is:
\begin{align*}
  & \fun{(s : \type{String})}                                                                    \\
  & \quad \doin{b}{\hilitecast{(\op[Random] \, \tmUnit)}{\dirty{\refl{\tyBool}}{\coerDirt^1}}}{} \\
  & \quad \big(\kpre{if} b \kop{then}                                                            \\
    & \quad\quad \op[Print] \, s                                                                   \\
    & \quad \kop{else}                                                                             \\
    & \quad\quad \hilitecast{(\return{\tmUnit})}{\dirty{\refl{\tyUnit}}{\coerDirt^2}}              \\
  & \quad \hilitecast{\big)}{\dirty{\refl{\tyUnit}}{\coerDirt^3}}
\end{align*}
Here $\dirty{\coerVal}{\coerDirt}$ is a computation coercion that witnesses subtyping $\dirty{\vty}{\dirt} \le \dirty{\vty'}{\dirt'}$ using $\coerVal : \vty \le \vty'$ for values and $\coerDirt : \dirt \le \dirt'$ for effects.
Note that even though effects are the only reason to include subtyping in \eff, we also need value coercions as higher-order functions need to accommodate arguments with differing effects.
In the example, $\refl{\tyBool} : \tyBool \le \tyBool$ and $\refl{\tyUnit} : \tyUnit \le \tyUnit$ are reflexive value coercions, while the dirt coercions are suitably constructed witnesses for:
\begin{align*}
  \coerDirt^1 : \set{\op[Random]} & \le \set{\op[Random], \op[Print]} \\
  \coerDirt^2 : \emptyset         & \le \set{\op[Print]}              \\
  \coerDirt^3 : \set{\op[Print]}  & \le \set{\op[Random], \op[Print]}
\end{align*}
So, $\coerDirt^2$ increases the effects of the otherwise pure second branch to the effect $\set{\op[Print]}$ of the first branch, while $\coerDirt^1$ and $\coerDirt^3$ increase the effects of both sequenced computations to their union $\set{\op[Random], \op[Print]}$.

\subsection{Compiling to a pure language}

The main advantage of using explicit witnesses for subtyping, rather than enforcing subtyping implicitly during inference, becomes evident at the compilation stage.

When targeting a functional language with no support for algebraic effects, effectful computations that yield a result of type $A$ can be represented using a user-defined type $\kpre{Comp} A$, which typically uses one of the known encodings, such as free monads, delimited control, or continuation-passing style.
While types $\dirty{\tyBool}{\set{\op[Random]}}$ or $\dirty{\tyBool}{\set{\op[Random], \op[Print]}}$ can both be compiled to $\kpre{Comp} \tyBool$, computations of type $\dirty{\tyBool}{\emptyset}$ are pure and do not require monadic encoding.
Instead, we want to compile them to (target language) computations of type $\tyBool$, gaining significant speed-up~\cite{DBLP:journals/pacmpl/KarachaliasKPS21}.

In \eff, subtyping coercions are used solely to ensure type compatibility and are discarded during evaluation.
More interestingly, though, they gain computational content during translation, as according to the above efficient translation, they are mapped to injections from $A$ to $\kpre{Comp} A$ at the appropriate points.
We do not go into the exact placement strategy here~\cite{DBLP:journals/jfp/KarachaliasPSVS20}, but it is worth noting that this process is subtle even for first-order computations, let alone higher-order ones, which is why having explicit witnesses is crucial.

For example the computation $\return{42} : \dirty{\tyInt}{\emptyset}$ should be translated simply to $42 : \tyInt$ in the target language.
On the other hand, a cast $\cast{(\return{42})}{\coerComp}$, where the coercion $\coerComp : \dirty{\tyInt}{\emptyset} \le \dirty{\tyInt}{\set{\op[Random]}}$ needs to end up in $\kpre{Comp} \tyInt$ and is thus translated as $return \, 42$.
Here, $return : A \to \kpre{Comp} A$ is the monadic injection in the target language (written as a variable to distinguish it from the keyword $\keyreturn$ in \eff).
To complicate things even further, a nested cast such as $\cast{(\cast{(\return{42})}{\coerComp})}{\coerComp'}$ where $\coerComp' : \dirty{\tyInt}{\set{\op[Random]}} \le \dirty{\tyInt}{\set{\op[Random], \op[Print]}}$ is still translated as $return \, 42$ and not as $return \, (return \, 42)$, which naive replacement of casts with $return$ would imply.

In implementation and in the following examples, we target \ocaml, as its syntactic similarity to \eff makes the translation particularly straightforward.
Even though version~5 of \ocaml introduced native support for handlers, our compilation approach remains relevant: \ocaml restricts handlers to continuations that may be resumed at most once, and many other target languages still lack native handler support.
When translating \eff computations to \ocaml, we target a monadic type \ocamlcode{'a comp}, which features functions such as \ocamlcode{performRandom : unit -> bool comp} for all the operations.
The above example would thus be translated as
\begin{ocamllisting}
let print_randomly s =
  performRandom () >>= fun b ->
    if b then performPrint s else return ()
\end{ocamllisting}
where \ocamlcode{(>>=) : 'a comp -> ('a -> 'b comp) -> 'b comp} is the monadic bind, and \ocamlcode{return : 'a -> 'a comp} is the monadic unit.

\subsection{Explicit polymorphism}
\label{sub:explicit-polymorphism}

To see the issues that polymorphism brings to our pipeline, let us ignore effect annotations for a bit.
Consider a polymorphic function that applies a given function~\effcode{f} to its argument~\effcode{x} only if a given predicate~\effcode{p} is satisfied.
In \eff, one would write it as:
\begin{efflisting}
let apply_if p f x = if p x then f x else x
\end{efflisting}
When translating the above function to \coreEff, the inference algorithm again annotates the resulting term with types and coercions, each now featuring parameters due to the presence of polymorphism:
\begin{align*}
  & \mathit{applyIf} = \fun{(p : \tyvar_1 \to \tyBool)} \keyreturn~\big( \fun{(f : \tyvar_2 \to \tyvar_3)} \keyreturn~( \fun{(x : \tyvar_4)} \\
      & \quad \doin{b}{p \, (\hilitecast{x}{\tycoervar_1})}{}                                                                                    \\
      & \quad \kpre{if} b \kop{then}                                                                                                             \\
      & \quad\quad \hilitecast{(f \, (\hilitecast{x}{\tycoervar_2}))}{\tycoervar_3}                                                              \\
      & \quad \kop{else}                                                                                                                         \\
      & \quad\quad \return(\hilitecast{x}{\tycoervar_4})                                                                                         \\
  & ) \big)
\end{align*}
In general, the type $\tyvar_4$ of $x$ does not have to be the same as the argument type $\tyvar_1$ of $p$, it only needs to be its subtype, thus it needs to be cast by some type coercion $\tycoervar_1 : \tyvar_4 \le \tyvar_1$.
A similar situation occurs with $f$ and $\tycoervar_2 : \tyvar_4 \le \tyvar_2$. Finally, the result of $f$ needs to be coerced to a supertype it shares with the other branch.
As before, we express this by introducing a new type parameter $\tyvar_5$ for the output and coercions $\tycoervar_3 : \tyvar_3 \le \tyvar_5$ and $\tycoervar_4 : \tyvar_4 \le \tyvar_5$.

The inferred type of the function is then:
\[
  \mathit{applyIf} :
  (\tyvar_1 \to \tyBool)
  \to
  (\tyvar_2 \to \tyvar_3)
  \to
  (\tyvar_4 \to \tyvar_5)
\]
under the constraints described above.
These are most easily expressed with a graph where nodes correspond to type parameters, while edges correspond to coercions between them:
\[
  \begin{tikzpicture}
    \node [shape=circle, draw] (t4) at (0.75, 0) {$\tyvar_4$};
    \node [shape=circle, draw] (t1) at (0, 2) {$\tyvar_1$};
    \node [shape=circle, draw] (t2) at (1.5, 2) {$\tyvar_2$};
    \node [shape=circle, draw] (t3) at (3, 0) {$\tyvar_3$};
    \node [shape=circle, draw] (t5) at (3, 2) {$\tyvar_5$};
    \draw[->] (t4) -- (t1) node[midway, left] {$\tycoervar_1$};
    \draw[->] (t4) -- (t2) node[midway, left] {$\tycoervar_2$};
    \draw[->] (t3) -- (t5) node[midway, left] {$\tycoervar_3$};
    \draw[->] (t4) -- (t5) node[midway, left] {$\tycoervar_4$};
  \end{tikzpicture}
\]

When translating into OCaml, we now also need to take coercion parameters into account.
Since their exact instantiations are known only at call-sites, we need to represent each parameter with an additional functional argument, so the (type-annotated) translation is:
\begin{ocamllisting}
let apply_if (w1 : 'a4->'a1) (w2 : 'a4->'a2) (w3 : 'a3->'a5)
  (w4 : 'a4->'a5) (p : 'a1->bool) (f : 'a2->'a3) (x : 'a4) : 'a5 =

  p (x |> w1) >>= fun b ->
    if b then (f (x |> w2)) |> coer_comp w3
    else return (x |> w4)
\end{ocamllisting}
where \ocamlcode{coer_comp : ('a -> 'b) -> ('a comp -> 'b comp)} lifts value coercions to computations,
and \ocamlcode{(|>) : 'a -> ('a -> 'b) -> 'b} is the reverse application operator.

Each additional argument not only increases the function size, but also significantly impacts the runtime, especially as it soon prevents the \ocaml compiler from passing function arguments in processor registers.
In fact, the situation is even worse, as the above example was simplified for the presentation, while the actual inference algorithm uses additional intermediate type parameters leading to further arguments.
A quick analysis of the \eff standard library shows that the number of explicit constraints produced by the current algorithm is overwhelming.
The standard \prog{unzip} function needs 23 explicit coercions, while a straightforward implementation of the quicksort algorithm produces \ocaml code with just under 200 explicit coercion parameters.

\subsection{Type constraint simplification}
\label{sec:overview-simplification}

Explicit type coercions are relevant during compilation, optimization and during execution, and therefore cannot be simply discarded, but it is obvious we want to reduce them to a minimum.
One simple case of a simplification is when the coercions form a cycle in the graph, which often happens with recursive functions.
For example, if we have three coercions:
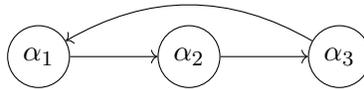
\[
  \begin{tikzpicture}
    \node [shape=circle, draw] (t1) at (0, 0) {$\tyvar_1$};
    \node [shape=circle, draw] (t2) at (2, 0) {$\tyvar_2$};
    \node [shape=circle, draw] (t3) at (4, 0) {$\tyvar_3$};
    \draw[->] (t1) -- (t2);
    \draw[->] (t2) -- (t3);
    \draw[->] (t3) edge[bend right] (t1);
  \end{tikzpicture}
\]
all the parameters must be the same, and we can replace them with a single parameter $\tyvar$.
Similarly, we can replace all coercions with the reflexive coercion $\refl{\tyvar} : \tyvar \le \tyvar$, essentially removing any casts in which they appear.

In a more general case, we have to be more cautious, though we can achieve similar results.
To see how constraints impact the type of the program, we must first examine the positions at which parameters occur in types.
For the sake of exposition, assume that our calculus features types $\tyInt \le \tyFloat$ (we are going to omit writing the coercion if we shall be interested only in the subtyping relation between the types).

Recall that subtyping of function types is \emph{contravariant} in their domain.
We have e.g. $(\tyBool \to \tyInt) \le (\tyBool \to \tyFloat)$ because every function returning integers can also act as one returning floats.
But on the other hand, we have $(\tyFloat \to \tyBool) \le (\tyInt \to \tyBool)$ since every function accepting floats can also accept integers, and not the other way around.

Increasing an instantiation of a type parameter $\tyvar$ in a covariant position (e.g. in $\tyBool \to \tyvar$, but also in $(\tyvar \to \tyBool) \to \tyBool$ where contravariance acts twice) thus increases the resulting instantiated type, while increasing an instantiation of a parameter in a contravariant position does the exact opposite.
To track the impact of parameters on a type, we assign them a \emph{polarity}: \emph{positive} for those that appear covariantly in the type, and \emph{negative} for those that appear contravariantly.

The running example above annotated with polarities is:
\[
  \mathit{applyIf} :
  (\tyvar_1^+ \to \tyBool)
  \to
  (\tyvar_2^+ \to \tyvar_3^-)
  \to
  (\tyvar_4^- \to \tyvar_5^+)
\]
\[
  \begin{tikzpicture}
    \node [shape=circle, draw] (t4) at (0.75, 0) {$\tyvar_4^-$};
    \node [shape=circle, draw] (t1) at (0, 2) {$\tyvar_1^+$};
    \node [shape=circle, draw] (t2) at (1.5, 2) {$\tyvar_2^+$};
    \node [shape=circle, draw] (t3) at (3, 0) {$\tyvar_3^-$};
    \node [shape=circle, draw] (t5) at (3, 2) {$\tyvar_5^+$};
    \draw[->] (t4) -- (t1) node[midway, left] {$\tycoervar_1$};
    \draw[->] (t4) -- (t2) node[midway, left] {$\tycoervar_2$};
    \draw[->] (t3) -- (t5) node[midway, left] {$\tycoervar_3$};
    \draw[->] (t4) -- (t5) node[midway, left] {$\tycoervar_4$};
  \end{tikzpicture}
\]
As in this example, it is often the case that each parameter is either positive or negative and that constraints are of the form $X^- \le Y^+$, however that does not hold in general!

Once parameters are assigned polarities, we can show that we lose no generality if we replace all positive parameters with their unique lower bounds, when they exist (and conversely for negative parameters and unique upper bounds).
For example, we can safely (proving that is one of the main results of this paper) replace our function with
\[
  \mathit{applyIf}' :
  (\tyvar_4^\pm \to \tyBool)
  \to
  (\tyvar_2^+ \to \tyvar_3^-)
  \to
  (\tyvar_4^\pm \to \tyvar_5^+)
\]
obtained by coalescing the positive parameter $\tyvar_1^+$ with its unique lower bound $\tyvar_4^-$ (which now gains positive polarity) and replacing all occurrences of $\tycoervar_1$ in the body with $\refl{\tyvar_4^-}$.
On the surface, $\mathit{applyIf}'$ appears less flexible than $\mathit{applyIf}$, but this can be amended at the call site.

Indeed, any application $\mathit{applyIf} \, p \, f \, v$, where we instantiated all type parameters $\tyvar_i$ with some types $\vty_i$ and all coercion parameters $\tycoervar_j$ with some coercions $\coercion_j$, can be replaced by an application $(\cast{\mathit{applyIf}'}{\coercion'}) \, p \, f \, v$ of a suitably coerced $\mathit{applyIf}'$, where $\coercion'$ is a witness for
\[
  \big((\vty_4 \to \tyBool)
    \to
    (\vty_2 \to \vty_3)
    \to
  (\vty_4 \to \vty_5)\big)
  \le
  \big((\vty_1 \to \tyBool)
    \to
    (\vty_2 \to \vty_3)
    \to
  (\vty_4 \to \vty_5)\big)
\]
which exists since $\coercion_1$ witnesses $\vty_4 \le \vty_1$ and $\tyvar_1^+$ was positive.

One could argue that we have replaced one coercion with an even more involved one, but this can often be optimized away due to the extra context information available at the call site.
And even if we cannot completely remove $\coercion'$, we have reduced the number of coercion parameters, which has a much more significant impact on performance.

We can continue the simplification by removing $\tycoervar_2$ and $\tycoervar_3$. We pick a particular order, though any other would lead to similar results:
\[
  \begin{tikzpicture}[baseline={(current bounding box.center)}]
    \node [shape=circle, draw] (t4) at (0.75, 0) {$\tyvar_4^-$};
    \node [shape=circle, draw] (t1) at (0, 2) {$\tyvar_1^+$};
    \node [shape=circle, draw] (t2) at (1.5, 2) {$\tyvar_2^+$};
    \node [shape=circle, draw] (t3) at (3, 0) {$\tyvar_3^-$};
    \node [shape=circle, draw] (t5) at (3, 2) {$\tyvar_5^+$};
    \draw[->, dotted] (t4) -- (t1) node[midway, left] {$\tycoervar_1$};
    \draw[->] (t4) -- (t2) node[midway, left] {$\tycoervar_2$};
    \draw[->] (t3) -- (t5) node[midway, left] {$\tycoervar_3$};
    \draw[->] (t4) -- (t5) node[midway, left] {$\tycoervar_4$};
  \end{tikzpicture}
  \quad\scalebox{1.5}{$\leadsto$}\quad
  \begin{tikzpicture}[baseline={(current bounding box.center)}]
    \node [shape=circle, draw] (t4) at (1.5, 0) {$\tyvar_4^\pm$};
    \node [shape=circle, draw] (t2) at (1.5, 2) {$\tyvar_2^+$};
    \node [shape=circle, draw] (t3) at (3, 0) {$\tyvar_3^-$};
    \node [shape=circle, draw] (t5) at (3, 2) {$\tyvar_5^+$};
    \draw[->, dotted] (t4) -- (t2) node[midway, left] {$\tycoervar_2$};
    \draw[->] (t3) -- (t5) node[midway, left] {$\tycoervar_3$};
    \draw[->] (t4) -- (t5) node[midway, left] {$\tycoervar_4$};
  \end{tikzpicture}
  \quad\scalebox{1.5}{$\leadsto$}\quad
  \begin{tikzpicture}[baseline={(current bounding box.center)}]
    \node [shape=circle, draw] (t4) at (1.5, 0) {$\tyvar_4^\pm$};
    \node [shape=circle, draw] (t3) at (3, 0) {$\tyvar_3^-$};
    \node [shape=circle, draw] (t5) at (3, 2) {$\tyvar_5^+$};
    \draw[->, dotted] (t3) -- (t5) node[midway, left] {$\tycoervar_3$};
    \draw[->] (t4) -- (t5) node[midway, left] {$\tycoervar_4$};
  \end{tikzpicture}
  \quad\scalebox{1.5}{$\leadsto$}\quad
  \begin{tikzpicture}[baseline={(current bounding box.center)}]
    \node [shape=circle, draw] (t4) at (0, 0) {$\tyvar_4^\pm$};
    \node [shape=circle, draw] (t5) at (0, 2) {$\tyvar_5^\pm$};
    \draw[->] (t4) -- (t5) node[midway, right] {$\tycoervar_4$};
  \end{tikzpicture}
\]
ending up with the final type
\[
  (\tyvar_4^\pm \to \tyBool)
  \to
  (\tyvar_4^\pm \to \tyvar_5^\pm)
  \to
  (\tyvar_4^\pm \to \tyvar_5^\pm)
\]
and the compiled \ocaml function with a single remaining coercion parameter:
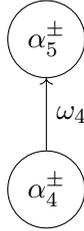
\begin{ocamllisting}
let apply_if w4 p f x =
  p x >>= fun b -> if b then f x else return (x |> w4)
\end{ocamllisting}
Note that we cannot simplify the type any further because $\alpha_4$ is positive and $\alpha_5$ is negative.

Again, even though the simplified type seems more restricted than the original one, polarity of parameters again ensures that for any instantiation $\{ (\tyvar_i \mapsto \vty_i)_i, (\tycoervar_j \mapsto \coercion_j)_j \}$, we have a witness
\[
  \big((\vty_4 \to \tyBool)
    \to
    (\vty_4 \to \vty_5)
    \to
  (\vty_4 \to \vty_5)\big)
  \le
  \big((\vty_1 \to \tyBool)
    \to
    (\vty_2 \to \vty_3)
    \to
  (\vty_4 \to \vty_5)\big)
\]
constructed using coercions $\coercion_1 : \vty_4 \le \vty_1$, $\coercion_2 : \vty_4 \le \vty_2$, and $\coercion_3 : \vty_3 \le \vty_5$.

In general, the order in which we apply the simplifications affects possible future optimisations, however this is in practice more an exception than a rule.
Our implementation, which chooses coercions in an arbitrary order, compiles the quicksort algorithm without any of the hundreds of additional parameters, and the same holds for the standard library of \eff, which contains around 50 of the most frequently used polymorphic functions (see~Section~\ref{sub:results}).

\subsection{Effect annotations}
\label{sub:effect-annotations}

Coming back to effect information, the annotated version of the running example becomes:
\begin{align*}
  & \fun{(p : \tyvar_1 \dirtyTo{\hilite{\dirtvar_1}} \tyBool)} \keyreturn~(                                             \\
    & \quad \fun{(f : \tyvar_2 \dirtyTo{\hilite{\dirtvar_2}} \tyvar_3)} \keyreturn~(                                      \\
      & \quad\quad \fun{(x : \tyvar_4)} (                                                                                   \\
        & \quad\quad\quad \doin{b}{\cast{(p \, (\cast{x}{\tycoervar_1}))}{\hilite{\dirty{\refl{\tyBool}}{\dirtcoervar_1}}}}{} \\
        & \quad\quad\quad \kpre{if} b \kop{then}                                                                              \\
        & \quad\quad\quad\quad \hilitecast{(f \, (\cast{x}{\tycoervar_2}))}{\dirty{\tycoervar_3}{\dirtcoervar_2}}             \\
        & \quad\quad\quad \kop{else}                                                                                          \\
        & \quad\quad\quad\quad \hilitecast{(\return{x})}{\dirty{\tycoervar_4}{\emptyset_{\dirtvar_3}}}                        \\
  & ) ) )
\end{align*}
Dirt parameters~$\dirtvar_1$ and $\dirtvar_2$ capture sets of effects performed by $p$ and $f$, respectively.
We do not need the effects of $p$ and $f$ to be the same, we just need the effect of the whole computation to cover both.
For that reason, we extend subtyping to dirt and capture the above by introducing a third dirt parameter $\dirtvar_3$ and two dirt coercion parameters $\dirtcoervar_1 : \dirtvar_1 \le \dirtvar_3$ and $\dirtcoervar_2 : \dirtvar_2 \le \dirtvar_3$.

In the example above, the first coercion casts the application of $p$ to $x$.
In it, we use $\dirtcoervar_1$ to cast the dirt, while for the type $\tyBool$ of boolean values, the only option is the reflexive coercion $\refl{\tyBool} : \tyBool \le \tyBool$.
Next, we cast the application of $f$ to $x$ by extending the value coercion $\tycoervar_3$ from before with the dirt coercion $\dirtcoervar_2$.
Finally, in the last cast, returned values cause no effects, so we need to cast the empty dirt to the common dirt $\dirtvar_3$, which we do using the empty coercion $\emptyset_{\dirtvar_3} : \emptyset \le \dirtvar_3$.

Even though dirt coercions only inform the monadic structure and get erased when translating to \ocaml, they are still present in the intermediate representation and other potential compilation targets, so we want to simplify them.
This proceeds similarly as for type coercions.
Incoming dirt parameters $\dirtvar_1$ and $\dirtvar_2$ receive negative polarity and outgoing dirt $\dirtvar_3$ receives the positive one.
Thus, the constraint graph is of the form:
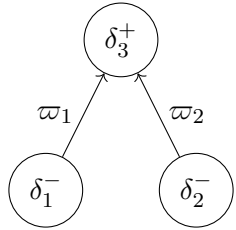
\[
  \begin{tikzpicture}
    \node [shape=circle, draw] (d1) at (6, 0) {$\dirtvar_1^-$};
    \node [shape=circle, draw] (d2) at (8, 0) {$\dirtvar_2^-$};
    \node [shape=circle, draw] (d3) at (7, 2) {$\dirtvar_3^+$};
    \draw[->] (d1) -- (d3) node[midway, left] {$\dirtcoervar_1$};
    \draw[->] (d2) -- (d3) node[midway, right] {$\dirtcoervar_2$};
  \end{tikzpicture}
\]
which we can simplify by replacing both $\dirtvar_1$ and $\dirtvar_2$ with $\dirtvar_3$ and both $\dirtcoervar_1$ and $\dirtcoervar_2$ with the reflexive coercion $\refl{\dirtvar_3}$.

With some renaming and writing $A \dirtyTo{\emptyset} B$ as $A \to B$, the final simplified type of \effcode{apply_if} is thus:
\[
  (\alpha \dirtyTo{\dirtvar} \tyBool)
  \to
  (\alpha \dirtyTo{\dirtvar} \beta)
  \to
  \alpha \dirtyTo{\dirtvar} \beta
\]
under the constraint $\alpha \le \beta$.
It is worth reiterating the simplified type is equivalent to the original type
\[
  (\tyvar_1 \dirtyTo{\dirtvar_1} \tyBool)
  \to
  (\tyvar_2 \dirtyTo{\dirtvar_2} \tyvar_3)
  \to
  (\tyvar_4 \dirtyTo{\dirtvar_3} \tyvar_5)
\]
under the constraints $\tyvar_3 \le \tyvar_5, \tyvar_4 \le \tyvar_1, \tyvar_4 \le \tyvar_2, \tyvar_4 \le \tyvar_5, \dirtvar_1 \le \dirtvar_3, \dirtvar_2 \le \dirtvar_3$.

\subsection{Operations}

To account for both operations and dirt parameters, the general form of a dirt is $\ops \cup \dirtvar$, where $\dirtvar$ is a \emph{row} parameter capturing all operations not listed in $\ops$.
For example, here is a program that applies a given function to its argument with some probability:
\begin{efflisting}
let apply_randomly f x =
  if (perform Random) then f x else x
\end{efflisting}
Simplification of type constraints removes all of them and the type of \effcode{apply_randomly} is:
\[
  (\tyvar_1 \dirtyTo{\dirtvar_1} \tyvar_2) \to \tyvar_1 \dirtyTo{\dirtvar_2} \tyvar_2
\]
under the constraints $\dirtvar_1 \le \dirtvar_2$ and $\set{\op[Random]} \le \dirtvar_2$.
The latter constraint shows that $\dirtvar_2$ needs to contain $\op[Random]$, so we replace it with $\set{\op[Random]} \cup \dirtvar_3$ for some fresh dirt parameter $\dirtvar_3$, ending up with a single constraint of the form
\[
  \dirtvar_1 \le \set{\op[Random]} \cup \dirtvar_3
\]
Furthermore, $\dirtvar_1$ is a negative parameter and can be safely increased to $\set{\op[Random]} \cup \dirtvar_3$, leading to the final simplified type
\[
  (\alpha \longDirtyTo{\set{\op[Random]} \cup \dirtvar_3} \beta) \to \alpha \longDirtyTo{\set{\op[Random]} \cup \dirtvar_3} \beta
\]
One may worry about possible inefficiency, since the functional argument could have been pure if $\dirtvar_1$ were set to $\emptyset$.
But now, its expanded dirt always contains $\op[Random]$ and is thus impure.
In our current compiler, we treat all dirt parameters as potentially impure, so further increasing $\dirtvar_1$ has no effect. However, this is something to keep in mind when specialising effect-polymorphic functions into pure and impure variants (see Section~\ref{sub:future-work}).

While adding handlers provides additional flexibility to the programmer, it does not change the effect system significantly~\cite{DBLP:journals/corr/BauerP13}.
To incorporate first-class handlers, we introduce a new type constructor $\dirty{\vty_1}{\dirt_1} \hto \dirty{\vty_2}{\dirt_2}$, which handles a computation of type $\dirty{\vty_1}{\dirt_1}$ into one of type $\dirty{\vty_2}{\dirt_2}$.
For example, a handler that resolves all \effcode{Random} operations to \effcode{true} and counts the number of \effcode{Random} operations performed can be written as:
\begin{efflisting}
handler
| val x -> (x, 0)
| perform Random k -> let (res, cnt) = k true in (res, cnt + 1)
\end{efflisting}
The type of this handler is
\[
  \dirty{\tyvar}{(\set{\op[Random]} \cup \dirtvar)} \hto \dirty{(\tyvar \times \tyInt)}{\dirtvar}
\]
So, it takes a probabilistic computation potentially performing $\op[Random]$ and some other effects $\dirtvar$ and producing values of type $\tyvar$, and handles it into a computation that performs only $\dirtvar$, and returns an integer in addition to the original result.

All of the above is more or less orthogonal to the rest of the work, so we will not go into further details.

\section{Language}
\label{sec:language}
Let us now turn to a formal development of the above ideas.
As our working language, we take \coreEff, a simple fine-grain call-by-value~\cite{DBLP:journals/iandc/LevyPT03} calculus with effects and explicit coercions.
In contrast to our previous work~\cite{DBLP:journals/jfp/KarachaliasPSVS20}, which used impredicative polymorphism, we use predicative polymorphism as it is simpler and sufficient for our use, especially as \eff generalizes only top-level let bindings~\cite{DBLP:conf/tldi/VytiniotisJS10}.

\subsection{Types}

As we are working with a fine-grain call-by-value calculus, we distinguish between value and computation types, defined as:
\[
  \begin{array}{r@{\quad}r@{~}c@{~}l}
    \text{dirt}             & \dirt & ::= &
    \dirtvar                                                       \mid
    \emptyset                                                      \mid
    \set{\op} \cup \dirt                    \\\\
    \text{value type}       & A, B  & ::= &
    \tyvar                                                         \mid
    \tyUnit                                                         \mid
    \vty \to \cty                           \\

    \text{computation type} & \cty  & ::= &
    \dirty{\vty}{\dirt}                     \\\\
    \text{skeleton}         & \skel & ::= &
    \skelvar                                                \mid
    \tyUnit                                              \mid
    \skel_1 \to \skel_2
  \end{array}
\]
Dirt is, as before, a sequence of operations followed by either a row dirt parameter or the empty dirt.
Note that the same set of operations~$\ops$ can be represented in multiple ways, all of which differ only in the ordering of operations.
In this paper, we consider all such representations equivalent and treat dirt as being of the form $\ops \cup \dirtvar$ in case the row parameter is present, or $\ops$ (i.e. $\ops \cup \emptyset$) if it is not.
The prototype implementation follows this convention as well, and represents dirt as an unordered set of operations and an optional dirt parameter, rather than an inductive sequence as defined by the above grammar.

For simplicity, value types are limited to type parameters, the unit type, and function types, which consist of an argument (value) type and a return (computation) type.
Since functions take value arguments and perform computations, their type consists of a value type domain and a computation type codomain.

Computation types are a combination of a type of returned values with an associated dirt representing the operations that may be called.
Importantly, the dirt conservatively over-approximates the set of operations that may be called.

For the rest of the paper, we assume a global signature
\[
  \Sigma = \set{\op_1 : A_1 \to B_1, \op_2: A_2 \to B_2, \dots, \op_n : A_n \to B_n }
\]
which assigns monomorphic types $A_i$ and $B_i$ to each operation $\op_i$.
Note that the arrow~$\to$ in the operation signature is just the traditional syntax separating two value types, and should not be confused with a function type where the right-hand side is a computation type.
To obtain a simpler denotational semantics (see Section~\ref{sub:effectful-semantics}), we assume the type $B_i$ to be first-order, i.e. not a function type.
Given that a similar but more involved semantics can be given for arbitrary types of returned values~\cite{DBLP:journals/corr/BauerP13}, we see no reason why our results would not hold in general.

In addition to types, we also introduce \emph{skeletons}~\cite{DBLP:journals/jfp/KarachaliasPSVS20}.
Our subtyping is concerned only with effect information and its direct impact on values via function types, but not with other forms of value subsumption, such as class inheritance or width and depth subtyping of record types.
For that reason, one value can be a subtype of another only if their type constructors match, sometimes also dubbed \emph{structural subtyping}~\cite{DBLP:conf/icfp/Pottier96,DBLP:journals/iandc/Pottier01,DBLP:conf/aplas/Simonet03} (not to be confused with structural subtyping as the opposite of the nominal one).
Skeletons concretely capture this shared structure and they behave like value types with all effect information erased.
For example, both $\tyUnit \to \dirty{\tyUnit}{\set{\op}}$ and $\tyUnit \to \dirty{\tyUnit}{\set{\op, \op[op']}}$ share the same skeleton $\tyUnit \to \tyUnit$.

Skeletons are a technical tool which becomes particularly useful in the presence of type parameters as instead of additionally tracking equivalence classes $\tyvar_1 \approx \dots \approx \tyvar_n$ of comparable type parameters~\cite{DBLP:conf/aplas/Simonet03,DBLP:journals/corr/Pretnar13}, one can simply assign them all the same skeleton~$\skel$ in some type parameter context $\tyvar_1 : \skel, \dots, \tyvar_n : \skel$.
In general, type parameter contexts and their skeleton and dirt counterparts (which are just lists of parameters), are given by the following grammar:
\[
  \begin{array}{r@{\quad}r@{~}c@{~}l}
    \text{skeleton parameter context} & \ctx_s & ::= & \emptyctx \mid \ctx_s, \skelvar         \\
    \text{dirt parameter context}     & \ctx_d & ::= & \emptyctx \mid \ctx_d, \dirtvar         \\
    \text{type parameter context}     & \ctx_t & ::= & \emptyctx \mid \ctx_t, (\tyvar : \skel)
  \end{array}
\]
Well-formedness judgements of contexts, skeletons and dirts is routine and given by rules in Figure~\ref{fig:wellformedness-boring} in Appendix~\ref{sec:wellFormedness}, while the slightly more interesting rules for types are given in Figure~\ref{fig:wellformedness-types}.

\begin{figure}[t]
  \begin{mathpar}
    \hilite{\tVty{\ctx_s;\ctx_d;\ctx_t}{\vty}{\skel}}\ \textrm{assuming} \ \tCtxn{s}{\ctx_s} \ \textrm{and} \ \tSty{\ctx_s}{\skel}\hfill
    \\
    \infer{(\tyvar : \skel) \in \ctx_t}{\tVty{\ctx_s;\ctx_d;\ctx_t}{\alpha}{\skel}}

    \infer{\tVty{\ctx_s;\ctx_d;\ctx_t}{\vty}{\skel_1} \\ \tCty{\ctx_s;\ctx_d;\ctx_t}{\cty}{\skel_2}}{\tVty{\ctx_s;\ctx_d;\ctx_t}{\vty \to \cty}{\skel_1 \to \skel_2}}

    \infer{
    }{\tVty{\ctx_s;\ctx_d;\ctx_t}{\tyUnit}{\tyUnit}}
    \\ 
    \hilite{\tCty{\ctx_s;\ctx_d;\ctx_t}{\cty}{\skel}}\ \textrm{assuming} \ \tCtxn{s}{\ctx_s} \ \textrm{and} \ \tSty{\ctx_s}{\skel}\hfill \\
    \infer{
      \tVty{\ctx_s;\ctx_d;\ctx_t}{\vty}{\skel} \\ \tDty{\ctx_d}{\dirt}
    }{
      \tCty{\ctx_s;\ctx_d;\ctx_t}{\dirty{\vty}{\dirt}}{\skel}
    }
  \end{mathpar}
  \caption{Well-formedness rules for types}
  \label{fig:wellformedness-types}
\end{figure}

\subsection{Coercions}

Coercions, denoted by $\coercion{}$ are explicit witnesses for the subtyping relation.
Similarly to types, coercions are split into ones for dirt, values, and computations.
\[
  \begin{array}{r@{\quad}r@{~}c@{~}l}
    \text{dirt coercion}        & \coerDirt & ::= &
    \dirtcoervar
    \mid %
    \coerDirt^2 \circ \coerDirt^1
    \mid 
    \refl{\dirtvar}
    \mid 
    \refl{\emptyset}
    \mid 
    \emptyset_\dirtvar
    \mid 
    \set{\op} \cupDirtBoth \coerDirt
    \mid 
    \set{\op} \cupDirtSingle \coerDirt
    \\ 
    \text{value coercion}       & \coerVal  & ::= &
    \tycoervar
    \mid
    \coerVal^2 \circ \coerVal^1
    \mid 
    \refl{\tyvar}
    \mid 
    \trgUnitRefl
    \mid 
    \coerVal \to \coerComp
    \\ 
    \text{computation coercion} & \coerComp & ::= &
    \dirty{\coerVal}{\coerDirt}                      
  \end{array}
\]
Dirt coercions can be either parameters, compositions, reflexive coercions for an arbitrary dirt parameter, reflexive empty dirt set coercion, coercion asserting that empty set is below any dirt parameter, or two extensions of coercions with operations.
The first one asserts that the same operation can be added on both sides, while the second one asserts that the right hand side can be safely increased.
Similarly to dirt, dirt coercions can be given with differing order of operations, which we do not distinguish.

Next, value coercions are again parameters, compositions, reflexive coercions for type parameters and the unit type, and function coercions~$\coerVal \to \coerComp$, which use $\coerVal$ to cast the argument and $\coerComp$ to cast the result.
Computation coercions are just a combination of value coercion and dirt coercion.

Coercions require us to introduce two additional kinds of parameter contexts:
\[
  \begin{array}{r@{\quad}r@{~}c@{~}l}
    \text{dirt coercion parameter context} & \ctx_{dc} & ::= & \emptyctx \mid \ctx_{dc}, (\dirtcoervar : \dirt_1 \le \dirt_2) \\
    \text{type coercion parameter context} & \ctx_{tc} & ::= & \emptyctx \mid \ctx_{tc}, (\tycoervar : \vty_1 \le \vty_2)
  \end{array}
\]
which are well-formed according to rules in Figure~\ref{fig:wellformedness-coercion-contexts} in Appendix~\ref{sec:wellFormedness}.

For readability, we shall write the quintuple $(\ctx_s, \ctx_d, \ctx_t, \ctx_{dc}, \ctx_{tc})$ as a single, flat \emph{parameter context}~$\ctx$, formed as:
\[
  \infer{
    \tCtxn{s}{\ctx_s} \\
    \tCtxn{d}{\ctx_d} \\
    \tCtxn[\ctx_s]{t}{\ctx_t} \\
    \tCtxn[\ctx_d]{dc}{\ctx_{dc}} \\
    \tCtxn[\ctx_s;\ctx_d;\ctx_t]{tc}{\ctx_{tc}}
  }{
    \tCtx{(\ctx_s, \ctx_d, \ctx_t, \ctx_{dc}, \ctx_{tc})}
  }
\]
We shall use the joint context even in judgements that require only particular sub-contexts.
For example, we shall write $\tVty{\ctx}{\vty}{\skel}$ instead of the more verbose $\tVty{\ctx_s;\ctx_d;\ctx_t}{\vty}{\skel}$ when $\ctx = (\ctx_s, \ctx_d, \ctx_t, \ctx_{dc}, \ctx_{tc})$.

Using the joint contexts, the well-formedness rules for coercions are given in Figure~\ref{fig:coercion-wellformedness}.
Note that the subtyping on value and computation types is structural, i.e. both sides need to have the same skeleton.

\begin{figure}[b]
  \begin{mathpar}
    \hilite{\tCoerDirt{\ctx}{\coerDirt}{\dirt \leq \dirt'}}\ \textrm{assuming} \ \tCtx{\ctx} \ \textrm{and} \ \tDty{\ctx}{\dirt} \ \textrm{and} \ \tDty{\ctx}{\dirt'} \hfill \\
    \infer{
      (\dirtcoervar : \dirt \leq \dirt') \in \ctx
    }{
      \tCoerDirt{\ctx}{\dirtcoervar}{\dirt \leq \dirt'}
    }

    \infer{
      \tCoerDirt{\ctx}{\coerDirt^1}{\dirt \leq \dirt'} \and
      \tCoerDirt{\ctx}{\coerDirt^2}{\dirt' \leq \dirt''}
    }{
      \tCoerDirt{\ctx}{\coerDirt^2 \circ \coerDirt^1}{\dirt \leq \dirt''}
    }

    \infer{
      \tDty{\ctx}{\dirtvar}
    }{\tCoerDirt{\ctx}{\refl{\dirtvar}}{\dirtvar \leq \dirtvar}}

    \infer{
    }
    {\tCoerDirt{\ctx}{\refl{\emptyset}}{\emptyset \leq \emptyset}}

    \infer{
    }
    {\tCoerDirt{\ctx}{\emptyset_\dirtvar}{\emptyset \leq \dirtvar}}

    \infer{
      \op \in \Sigma \\ \tCoerDirt{\ctx}{\coerDirt}{\dirt \leq \dirt'}
    }{\tCoerDirt{\ctx}{\set{\op} \cupDirtBoth \coerDirt}{\set{\op} \cup \dirt \leq \set{\op} \cup \dirt'}}

    \infer{
      \op \in \Sigma \\ \tCoerDirt{\ctx}{\coerDirt}{\dirt \leq \dirt'}
    }{\tCoerDirt{\ctx}{\set{\op} \cupDirtSingle \coerDirt}{\dirt \leq \set{\op} \cup \dirt'}} \\
    \hilite{\tCoerVal{\ctx}{\coerVal}{\vty \leq \vty'}}\ \textrm{assuming} \ \tCtx{\ctx} \ \textrm{and} \ \tSty{\ctx}{\skel} \ \textrm{and}  \ \tVty{\ctx}{\vty}{\skel} \ \textrm{and} \ \tVty{\ctx}{\vty'}{\skel} \hfill \\
    \infer{(\tycoervar : \vty \leq \vty') \in \ctx}{\tCoerVal{\ctx}{\tycoervar}{\vty \leq \vty'}}

    \infer{
      \tCoerVal{\ctx}{\coerVal^1}{\vty \leq \vty'} \and
      \tCoerVal{\ctx}{\coerVal^2}{\vty' \leq \vty''}
    }{
      \tCoerVal{\ctx}{\coerVal^2 \circ \coerVal^1}{\vty \leq \vty''}
    }

    \infer{\tVty{\ctx}{\alpha}{\skel}}{\tCoerVal{\ctx}{\refl{\tyvar}}{\tyvar \leq \tyvar}} \\

    \infer{{}}{\tCoerVal{\ctx}{\refl{\tyUnit}}{\tyUnit \leq \tyUnit}}

    \infer{
      \tCoerVal{\ctx}{\coerVal}{\vty' \leq \vty} \\ \tCoerComp{\ctx}{\coerComp}{\cty \leq \cty'}
    }{
      \tCoerVal{\ctx}{\coerVal \to \coerComp}{(\vty \to \cty) \leq  (\vty' \to \cty')}
    } \\
    \hilite{\tCoerComp{\ctx}{\coerComp}{\cty \leq \cty'}}\ \textrm{assuming} \ \tCtx{\ctx} \ \textrm{and} \ \tSty{\ctx}{\skel} \ \textrm{and} \ \tCty{\ctx}{\cty}{\skel} \ \textrm{and} \ \tCty{\ctx}{\cty'}{\skel} \hfill \\
    \infer{
      \tCoerVal{\ctx}{\coerVal}{\vty \leq \vty'} \\ \tCoerDirt{\ctx}{\coerDirt}{\dirt \leq \dirt'}
    }{
      \tCoerComp{\ctx}{\dirty{\coerVal}{\coerDirt}}{(\dirty{\vty}{\dirt}) \leq (\dirty{\vty'}{\dirt'})}
    }
  \end{mathpar}
  \caption{Coercion well-formedness rules}
  \label{fig:coercion-wellformedness}
\end{figure}

Even though we assume empty coercions only for the empty dirt and dirt parameters, they are admissible for any well-formed $\tDty{\ctx}{\dirt}$.
For those, we can define the empty coercion $\tCoerDirt{\ctx}{\emptyset_\dirt}{\emptyset \le \dirt}$ by
\[
  \emptyset_\emptyset = \refl{\emptyset}
  \qquad
  \emptyset_{\set{\op} \cup \dirt} = \set{\op} \cupDirtBoth \emptyset_\dirt
\]
Similarly, reflexive coercions are admissible for an arbitrary dirt $\tCoerDirt{\ctx}{\refl{\dirt}}{\dirt \le \dirt}$ by:
\[
  \refl{\{\op\} \cup \dirt } = \set{\op} \cupDirtBoth \refl{\dirt}
\]
and for an arbitrary value or computation type by:
\[
  \refl{\vty \to \cty} = \refl{\vty} \to \refl{\cty}
  \qquad
  \refl{\dirty{\vty}{\dirt}} = \dirty{\refl{\vty}}{\refl{\dirt}}
\]


Due to composition, we have multiple (semantically equivalent) coercions between two types, for example we can always compose a coercion with the reflexive one.
In our previous work~\cite{DBLP:journals/jfp/KarachaliasPSVS20}, we preferred to keep the calculus minimal, but Proposition~\ref{prop:substitutionTransitivity} requires compositions of arbitrary coercions, including parametric ones, therefore we add composition as an additional construct.

\subsection{Terms}

Last, and in fact least, we define the terms of the language.
It turns out that it does not matter what terms we choose as long as typing judgements are preserved by substitutions (Theorem~\ref{thm:substitutionPreservesJudgements}) and their denotational semantics respects subtyping (Theorem~\ref{th:sqCommute}).
The values are usual: variable, unit, function and also an explicit cast of a value.
Computations are: return that lifts a value to a computation, operation calls, sequencing, application, and explicit casts of computations.

\[
  \begin{array}{r@{\quad}r@{~}c@{~}l}
    \text{typing context} & \tyCtx & ::= &
    \emptyctx \mid
    \tyCtx, (x : \vty)                     \\\\

    \text{value}          & v      & ::= &
    x                                                              \mid
    \tmUnit                                                        \mid
    \fun{(x:\vty)}{c}                                              \mid
    \cast{v}{\coerVal}                     \\

    \text{computation}    & c      & ::= &
    \return{v}                                                     \mid
    \operation{v}{y : \vty}{c}                                     \mid
    \doin{x}{c_1}{c_2}                                             \mid
    v_1~v_2                                                        \mid
    \cast{c}{\coerComp}
  \end{array}
\]

Typing contexts are well-formed in a parameter context according to rules in Figure~\ref{tab:contextWellFormedness} in Appendix~\ref{sec:wellFormedness}, while terms are typed according to rules in Figure~\ref{tab:poly:typing}.
As typing contexts are secondary in our development, we shall always use their full name, and simply use \emph{context} to talk about parameter contexts~$\ctx$.

\begin{figure}[h]
  \begin{mathpar}
    \hilite{\tVal{\ctx}{\tyCtx}{v}{\vty}}\ \textrm{assuming} \ \tTyCtx{\ctx}{\tyCtx} \ \textrm{and} \ \tVty{\ctx}{\vty}{\skel} \hfill \\
    \infer{x : \vty \in \tyCtx}{\tVal{\ctx}{\tyCtx}{x}{\vty}}

    \infer{
    }{\tVal{\ctx}{\tyCtx}{\tmUnit}{\tyUnit}}

    \infer{
      \tComp{\ctx}{\tyCtx, x: \vty}{c}{\cty}
    }{\tVal{\ctx}{\tyCtx}{\fun{(x : \vty)}{c}}{\vty \to \cty}}

    \infer{
      \tVal{\ctx}{\tyCtx}{v}{\vty} \\ \tCoerVal{\ctx}{\coerVal}{\vty \leq \vty'}
    }{\tVal{\ctx}{\tyCtx}{\cast{v}{\coerVal}}{\vty'}} \\

    \hilite{\tComp{\ctx}{\tyCtx}{c}{\cty}}\ \textrm{assuming} \ \tTyCtx{\ctx}{\tyCtx} \ \textrm{and} \ \tCty{\ctx}{\cty}{\skel} \hfill \\

    \infer{\tVal{\ctx}{\tyCtx}{v}{\vty}}{\tComp{\ctx}{\tyCtx}{\return v}{\dirty{\vty}\emptyset}}

    \infer{
      \tComp{\ctx}{\tyCtx}{c_1}{\dirty{\vty}{\dirt}} \\ \tComp{\ctx}{\tyCtx, x : \vty}{c_2}{\dirty{\vty'}{\dirt}}
    }{\tComp{\ctx}{\tyCtx}{\doin{x}{c_1}{c_2}}{\dirty{\vty'}{\dirt}}}

    \infer{\tVal{\ctx}{\tyCtx}{v_1}{\vty \to \cty} \\ \tVal{\ctx}{\tyCtx}{v_2}{\vty}}{\tComp{\ctx}{\tyCtx}{v_1~v_2}{\cty}}

    \infer{\op : \vty_1 \to \vty_2 \in \Sigma \\ \tVal{\ctx}{\tyCtx}{v}{\vty_1} \\ \tComp{\ctx}{\tyCtx, y : \vty_2}{c}{\dirty{\vty}{\dirt}} \\ \op \in \dirt}{\tComp{\ctx}{\tyCtx}{\operation{v}{y : \vty_2}{c}}{\dirty{\vty}{\dirt}}}

    \infer{\tComp{\ctx}{\tyCtx}{c}{\cty} \\ \tCoerComp{\ctx}{\coerComp}{\cty \leq \cty'}}{\tComp{\ctx}{\tyCtx}{\cast{c}{\coerComp}}{\cty'}}
  \end{mathpar}
  \caption{\coreEff typing rules}
  \label{tab:poly:typing}
\end{figure}

\section{Substitutions}
\label{sec:substitutions}
Most of our work focuses on translating types and terms in a given context into ones in a simpler context.
With that purpose, we define \emph{substitutions}~$\sigma$ as:
\[
  \begin{array}{r@{\quad}r@{~}c@{~}l}
    \text{substitution} & \sigma & ::= &
    \emptysub \mid
    \sigma, (\skelvar \mapsto \skel) \mid
    \sigma, (\dirtvar \mapsto \dirt) \mid
    \sigma, (\tyvar \mapsto \vty) \mid
    \sigma, (\tycoervar \mapsto \coerVal) \mid
    \sigma, (\dirtcoervar \mapsto \coerDirt)
  \end{array}
\]
that maps skeleton parameters to skeletons, dirt parameters to dirt, and so on.
We assume that each parameter is assigned at most one counterpart, so we can interpret each substitution as a function.
When easier, we will use the function like notation (e.g. $\sigma(\skelvar) = \skel$) to denote the action of the substitution on a given parameter.
We shall also use $\sigma \setminus \mathcal{P}$ to denote the substitution $\sigma$ with the removal of all parameters in the set $\mathcal{P}$.

\subsection{Valid substitutions}

A \emph{valid substitution} $\validSub{\sigma}{\ctx}{\ctx'}$ takes each parameter from a well-formed parameter context $\ctx$ into its counterpart, which is well-formed in context $\ctx'$, as given by rules presented in the Figure~\ref{fig:poly:substitution}.

\begin{figure}[h]
  \begin{mathpar}
    \hilite{\validSub{\sigma}{\ctx}{\ctx'}}\ \textrm{assuming} \ \tCtx{\ctx} \ \textrm{and} \ \tCtx{\ctx'}\hfill \\
    \infer{{}}{\validSub{\emptysub}{\emptyctx}{\ctx'}}

    \infer{\validSub{\sigma}{\ctx}{\ctx'} \\ \tSty{\ctx'}{\skel} }{
      \validSub{(\sigma, \skelvar \mapsto \skel)}{(\ctx, \skelvar)}{\ctx'}
    }

    \infer{\validSub{\sigma}{\ctx}{\ctx'} \\ \tDty{\ctx'}{\dirt} }{
      \validSub{(\sigma, \dirtvar \mapsto \dirt)}{(\ctx, \dirtvar)}{\ctx'}
    }

    \infer{
      \validSub{\sigma}{\ctx}{\ctx'} \\ \tVty{\ctx'}{\vty}{\sigma(\skel)}
    }{
      \validSub{(\sigma, \tyvar \mapsto \vty)}{(\ctx, (\tyvar: \skel))}{\ctx'}
    }

    \infer{
      \validSub{\sigma}{\ctx}{\ctx'} \\ \tCoerty{\ctx'}{\coerVal}{\sigma(\vty_1) \leq \sigma(\vty_2)}
    }{
      \validSub{(\sigma, \tycoervar \mapsto \coerVal)}{(\ctx, (\tycoervar : \vty_1 \leq \vty_2))}{\ctx'}
    }

    \infer{\validSub{\sigma}{\ctx}{\ctx'} \\
      \tCoerty{\ctx'}{\coerDirt}{\sigma(\dirt_1) \leq \sigma(\dirt_2)}
    }{
      \validSub{(\sigma, \varpi \mapsto \coerDirt)}{(\ctx, (\dirtcoervar : \dirt_1 \leq \dirt_2))}{\ctx'}
    }
  \end{mathpar}
  \caption{Substitution validity rules}
  \label{fig:poly:substitution}
\end{figure}

Note that validity rules ensure that substitutions are defined on all parameters in the source context, even if they map back to themselves.
To prevent unnecessary clutter, we will omit these trivial mappings.

We will mostly be interested in substitutions $\validSub{\sigma}{\ctx}{\ctx'}$ where $\ctx'$ is a subset of $\ctx$, and $\sigma$ preserves all their common parameters.
In such cases, we shall call $\sigma$ a \emph{strengthening} of $\ctx$ and denote $\ctx'$ as $\sigma(\ctx)$.
For example
\[
  \set{\skelvar' \mapsto \tyUnit, \dirtvar' \mapsto \dirtvar, \tyvar \mapsto \tyUnit}
\]
(where we omit $\skelvar \mapsto \skelvar, \dirtvar \mapsto \dirtvar$),
strengthens $\ctx = \skelvar, \skelvar', \dirtvar, \dirtvar', (\tyvar : \skelvar')$ to $\ctx' = \skelvar, \dirtvar$.

As expected, valid substitutions preserve well-formedness and typing judgements, where the action of a substitution is extended from parameters to arbitrary types in Figure~\ref{fig:poly:substitution:types}, arbitrary coercions in Figure~\ref{fig:poly:substitution:coercions}, and arbitrary typing contexts and terms in Figure~\ref{fig:poly:substitution:terms}.

\begin{figure}[h]
  \begin{align*}
    \sigma(\tyUnit)            & = \tyUnit                            & \sigma(\skel_1 \to \skel_2)  & = \sigma(\skel_1) \to \sigma(\skel_2) \\
    \\
    \sigma(\emptyset)          & = \emptyset                          & \sigma(\set{\op} \cup \dirt) & = \set{\op} \cup \sigma(\dirt)        \\
    \\
    \sigma(\tyUnit)            & = \tyUnit                            & \sigma(\vty \to \cty)        & = \sigma(\vty) \to \sigma(\cty)       \\
    \sigma(\dirty{\vty}{\cty}) & = \dirty{\sigma(\vty)}{\sigma(\cty)}
  \end{align*}
  \caption{Substitution action on skeletons, dirt, and types}
  \label{fig:poly:substitution:types}
\end{figure}

\begin{figure}
  \begin{align*}
    \sigma(\refl{\dirtvar})                  & = \refl{\sigma(\dirtvar)}                  & \sigma(\refl{\emptyset})                   & = \refl{\emptyset}                            \\
    \sigma(\emptyset_\dirtvar)               & = \emptyset_{\sigma(\dirtvar)}             & \sigma(\set{\op} \cupDirtSingle \coerDirt) & = \set{\op} \cupDirtSingle \sigma(\coerDirt)  \\
    \sigma(\set{\op} \cupDirtBoth \coerDirt) & = \set{\op} \cupDirtBoth \sigma(\coerDirt)                                                                                              \\
    \\
    \sigma(\trgUnitRefl)                     & = \trgUnitRefl                             & \sigma(\refl{\tyvar})                      & = \refl{\sigma(\tyvar)}                       \\
    \sigma(\coerVal \to \coerComp)           & = \sigma(\coerVal) \to \sigma(\coerComp)   & \sigma(\dirty{\coerVal}{\coerDirt})        & = \dirty{\sigma(\coerVal)}{\sigma(\coerDirt)} \\
  \end{align*}
  \caption{Substitution action on dirt and type coercions}
  \label{fig:poly:substitution:coercions}
\end{figure}

\begin{figure}
  \begin{align*}
    \sigma(\emptyctx)           & = \emptyctx                           & \sigma(\tyCtx, x : \vty)          & = \sigma(\tyCtx), x : \sigma(\vty)                 \\
    \\
    \sigma(x)                   & = x                                   & \sigma(\tmUnit)                   & = \tmUnit                                          \\
    \sigma(\fun{x : \vty}{c})   & = \fun{x : \sigma(\vty)}{\sigma(c)}   & \sigma(\cast{v}{\coerVal})        & = \cast{\sigma(v)}{\sigma(\coerVal)}               \\
    \\
    \sigma(\return{v})          & = \return{\sigma(v)}                  & \sigma(\doin{x}{v}{c})            & = \doin{x}{\sigma(v)}{\sigma(c)}                   \\
    \sigma(v_1~v_2)             & = \sigma(v)~\sigma(w)                 & \sigma(\operation{v}{y: \vty}{c}) & = \operation{\sigma(v)}{y:\sigma(\vty)}{\sigma(c)} \\
    \sigma(\cast{v}{\coerComp}) & = \cast{\sigma(v)}{\sigma(\coerComp)}
  \end{align*}
  \caption{Substitution action on typing contexts and terms}
  \label{fig:poly:substitution:terms}
\end{figure}

\begin{prop}
  \label{prop:substitutionPreservesSkeletons}
  For an arbitrary valid substitution $\validSub{\sigma}{\ctx}{\ctx'}$, the following holds:
  \begin{itemize}
    \item if $\tSty{\ctx}{\skel}$, then $\tSty{\ctx'}{\sigma(\skel)}$,
    \item if $\tDty{\ctx}{\dirt}$, then $\tDty{\ctx'}{\sigma(\dirt)}$.
  \end{itemize}
\end{prop}

\begin{prop}
  \label{prop:substitutionPreservesTypes}
  For an arbitrary valid substitution $\validSub{\sigma}{\ctx}{\ctx'}$, the following holds:
  \begin{itemize}
    \item if $\tVty{\ctx}{\vty}{\skel}$, then $\tVty{\ctx'}{\sigma(\vty)}{\sigma(\skel)}$,
    \item if $\tCty{\ctx}{\cty}{\skel}$, then $\tCty{\ctx'}{\sigma(\cty)}{\sigma(\skel)}$.
  \end{itemize}
\end{prop}

\begin{prop}
  \label{prop:substitutionPreservesCoercions}
  For an arbitrary valid substitution $\validSub{\sigma}{\ctx}{\ctx'}$, the following holds:
  \begin{itemize}
    \item if $\tCoerVal{\ctx}{\coerVal}{\vty_1 \leq \vty_2}$, then $\tCoerVal{\ctx'}{\sigma(\coerVal)}{\sigma(\vty_1) \leq \sigma(\vty_2)}$,
    \item if $\tCoerComp{\ctx}{\coerComp}{\cty_1 \leq \cty_2}$, then $\tCoerComp{\ctx'}{\sigma(\coerComp)}{\sigma(\cty_1) \leq \sigma(\cty_2)}$,
    \item if $\tCoerDirt{\ctx}{\coerDirt}{\dirt_1 \leq \dirt_2}$, then $\tCoerDirt{\ctx'}{\sigma(\coerDirt)}{\sigma(\dirt_1) \leq \sigma(\dirt_2)}$.
  \end{itemize}
\end{prop}

\begin{thm}
  \label{thm:substitutionPreservesJudgements}
  For an arbitrary valid substitution $\validSub{\sigma}{\ctx}{\ctx'}$, the following holds:
  \begin{itemize}
    \item if $\tVal{\ctx}{\tyCtx}{v}{\vty}$, then $\tVal{\ctx'}{\sigma(\tyCtx)}{\sigma(v)}{\sigma(\vty)}$,
    \item if ${\tComp{\ctx}{\tyCtx}{c}{\cty}}$, then $\tComp{\ctx'}{\sigma(\tyCtx)}{\sigma(c)}{\sigma(\cty)}$.
  \end{itemize}
\end{thm}

For substitutions $\sigma$ and $\sigma'$, we can define their \emph{composition} as:
\begin{align*}
  (\sigma' \circ \sigma)(\skelvar)     & = \sigma'(\sigma(\skelvar))     &
  (\sigma' \circ \sigma)(\tyvar)       & = \sigma'(\sigma(\tyvar))         \\
  (\sigma' \circ \sigma)(\dirtvar)     & = \sigma'(\sigma(\dirtvar))     &
  (\sigma' \circ \sigma)(\tycoervar)   & = \sigma'(\sigma(\tycoervar))     \\
  (\sigma' \circ \sigma)(\dirtcoervar) & = \sigma'(\sigma(\dirtcoervar))
\end{align*}

\begin{prop}
  \label{prop:substitution-composition}
  Let $\validSub{\sigma}{\ctx}{\ctx'}$ and $\validSub{\sigma'}{\ctx'}{\ctx''}$ be valid substitutions. Then the composition of the two substitutions is a valid substitution $\validSub{\sigma' \circ \sigma}{\ctx}{\ctx''}$.
\end{prop}

\subsection{A canonical form of contexts}
\label{sub:reduced-contexts}

To ease further analysis, we first reduce contexts to a canonical form, following similar ideas as our previous unification algorithms~\cite{DBLP:journals/corr/Pretnar13,DBLP:conf/esop/SalehKPS18,DBLP:journals/jfp/KarachaliasPSVS20}.
Recall that parameter contexts are of the form $(\ctx_s, \ctx_d, \ctx_t, \ctx_{dc}, \ctx_{tc})$, where:
\begin{itemize}
  \item $\ctx_s$ lists skeleton parameters $\skelvar$,
  \item $\ctx_d$ lists dirt parameters $\dirtvar$,
  \item $\ctx_t$ lists type parameters and their skeletons $\tyvar : \skel$,
  \item $\ctx_{dc}$ lists dirt coercion parameters $\dirtcoervar : \dirt_1 \le \dirt_2$, and
  \item $\ctx_{tc}$ lists type coercion parameters $\tycoervar : \vty_1 \le \vty_2$.
\end{itemize}

Contexts $\ctx_s$ and $\ctx_d$ are just lists of parameters, so we leave them as is.
However, due to structural subtyping, the type parameters $\ctx_t$ can be reduced to ones of the form $\tyvar : \skelvar$, so with only parameter skeletons.
For example, if we have $(\tyvar : \skel_1 \to \skel_2) \in \ctx_t$, the only closed types we can assign to $\tyvar$ are function types. Thus, we can introduce two type parameters $\tyvar_1 : \skel_1, \tyvar_2 : \skel_2$, and replace all occurrences of $\tyvar$ with $\tyvar_1 \to \tyvar_2$.

With that in mind, we define an auxiliary function $\subPhase_t(\ctx_t; \ctx_t', \ctx_d', \sigma)$ that processes a context~$\ctx_t$ one parameter at a time and accumulates canonical type parameters $\ctx_t'$, additionally generated dirt parameters $\ctx_d'$, and a reducing substitution $\sigma$:
\begin{align*}
  \subPhase_t(\emptyctx; \ctx_t', \ctx_d', \sigma)
  & := \tuple{\ctx_t', \ctx_d', \sigma}                                                                                                                                         \\
  \subPhase_t((\ctx_t, \tyvar : \skelvar); \ctx_t', \ctx_d', \sigma)
  & := \subPhase_t(\ctx_t; (\ctx_t', \tyvar : \skelvar), \ctx_d', \sigma)                                                                                                       \\
  \subPhase_t((\ctx_t, \tyvar : \tyUnit); \ctx_t', \ctx_d', \sigma)
  & := \subPhase_t(\ctx_t; \ctx_t', \ctx_d', \set{\tyvar \mapsto \tyUnit} \circ \sigma)                                                                                         \\
  \subPhase_t((\ctx_t, \tyvar : \skel_1 \to \skel_2); \ctx_t', \ctx_d', \sigma)
  & :=
  \begin{multlined}[t]
    \subPhase_t((\ctx_t, \tyvar_1 : \skel_1, \tyvar_2 : \skel_2); \ctx_t', (\ctx_d', \dirtvar), \set{\tyvar \mapsto \tyvar_1 \dirtyTo{\dirtvar} \tyvar_2} \circ \sigma) \\
    \shoveleft[6.7cm]{(\text{$\dirtvar, \tyvar_1, \tyvar_2$ fresh})}
  \end{multlined}
\end{align*}
Note that even though the number of variables in the unreduced context~$\ctx_t$ increases in the last case, the total size of accompanying skeletons decreases, ensuring termination.

After reducing type parameters, we turn to type constraints, as reducing those may introduce further dirt constraints.
Using a similar approach as for types, we reduce type constraints to ones witnessing $\tyvar_1 \leq \tyvar_2$, so only inequalities between type parameters.
Like before, we define a function $\subPhase_{tc}(\ctx_{tc}; \ctx_{tc}', \ctx_{dc}', \sigma)$ that processes~$\ctx_{tc}$ and accumulates reduced type constraints $\ctx_{tc}'$, additionally generated dirt constraints $\ctx_{dc}'$, and a substitution $\sigma$ like before:
\begin{align*}
  \subPhase_{tc}(\emptyctx; \ctx_{tc}', \ctx_{dc}', \sigma)
  & := \tuple{\ctx_{tc}', \ctx_{dc}', \sigma}                                                                                                                                                                            \\
  \subPhase_{tc}((\ctx_{tc}, \tycoervar : \tyvar_1 \leq \tyvar_2); \ctx_{tc}', \ctx_{dc}', \sigma)
  & := \subPhase_{tc}(\ctx_{tc}; (\ctx_{tc}', \tycoervar : \tyvar_1 \leq \tyvar_2), \ctx_{dc}', \sigma)                                                                                                                  \\
  \subPhase_{tc}((\ctx_{tc}, \tycoervar : \tyUnit \leq \tyUnit); \ctx_{tc}', \ctx_{dc}', \sigma)
  & := \subPhase_{tc}(\ctx_{tc}; \ctx_{tc}', \ctx_{dc}', \set{\tycoervar \mapsto \refl{\tyUnit}} \circ \sigma)                                                                                                           \\
  \subPhase_{tc}((\ctx_{tc}, \tycoervar : \vty_1 \dirtyTo{\dirt} \vty_2 \leq \vty_1' \dirtyTo{\dirt'} \vty_2'); \ctx_{tc}', \ctx_{dc}', \sigma)
  & :=                                                                                                                                                                                                                   \\
  &
  \begin{multlined}
    \hspace{-7cm}
    \subPhase_{tc}((\ctx_{tc}, \tycoervar_1 : \vty'_1 \leq \vty_1, \tycoervar_2: \vty_2 \leq \vty'_2); \ctx_{tc}', (\ctx_{dc}', \dirtcoervar: \dirt \leq \dirt'), \set{\tycoervar \mapsto (\tycoervar_1 \dirtyTo{\dirtcoervar} \tycoervar_2)} \circ \sigma) \\
    \text{where $\tycoervar_1, \tycoervar_2, \dirtcoervar$ fresh}
  \end{multlined}
\end{align*}
Since well-formedness rules ensure that both sides of type inequalities share the same skeleton, and since the previous stage removed all type parameters with non-parameter skeletons, these four clauses cover all possible cases.

We can visualise reduced contexts with directed graphs.
First, we have a graph representing type coercions.
Its nodes are type parameters, and its edges are type coercion parameters between them.
As we can compare only types with the same skeleton, the graph decomposes into separate subgraphs, each corresponding to a single skeleton parameter.
For example, the reduced context
\begin{multline*}
  \skelvar_1, \skelvar_2,
  (\tyvar_1 : \skelvar_1), (\tyvar_1' : \skelvar_1), (\tyvar_1'' : \skelvar_1),
  (\tyvar_2 : \skelvar_2), (\tyvar_2' : \skelvar_2), (\tyvar_2'' : \skelvar_2) \\
  (\tycoervar_1 : \tyvar_1 \le \tyvar_1''),
  (\tycoervar_2 : \tyvar_1' \le \tyvar_1''),
  (\tycoervar_3 : \tyvar_2 \le \tyvar_2')
\end{multline*}
can be concisely captured with the graph
\[
  \begin{tikzpicture}[box/.style = {draw, dashed, rounded corners=2mm}]
    \node [shape=circle, draw] (a1) at (0, 0) {$\tyvar_1$};
    \node [shape=circle, draw] (a1') at (2, 0) {$\tyvar_1'$};
    \node [shape=circle, draw] (a1'') at (1, 2) {$\tyvar_1''$};
    \draw[->] (a1) edge node[midway, left] {$\tycoervar_1$} (a1'');
    \draw[->] (a1') edge node[midway, right] {$\tycoervar_2$} (a1'');
    \node[box, fit=(a1)(a1')(a1''), label=below:$\skelvar_1$] (s1) {};
    \node [shape=circle, draw] (a2) at (4, 0) {$\tyvar_2$};
    \node [shape=circle, draw] (a2') at (4, 2) {$\tyvar_2'$};
    \node [shape=circle, draw] (a2'') at (5, 1) {$\tyvar_2''$};
    \draw[->] (a2) edge node[midway, left] {$\tycoervar_3$} (a2');
    \node[box, fit=(a2)(a2')(a2''), label=below:$\skelvar_2$] (s2) {};
  \end{tikzpicture}
\]

Similarly, we shall reduce dirt coercions.
Unlike type coercions, we can reduce only the left-hand sides~\cite{DBLP:journals/jfp/KarachaliasPSVS20}, ending up with inequalities of the form $\dirtvar_1 \le \ops$ and $\dirtvar_1 \le \ops \cup \dirtvar_2$.
Still, we end up with a graph that has dirt parameters as nodes and dirt coercion parameters as edges between them.
Additionally, each edge is annotated with an operation set to account for the right-hand side of dirt inequalities, and we need an additional sink node to account for coercions with no parameter on the right.
In Section~\ref{sub:removing-parallel-edges}, we shall explain why we can put operations on edges rather than add additional nodes for each combination $\ops \cup \dirtvar$.
For example, the reduced context
\[
  \dirtvar_1, \dirtvar_2, \dirtvar_3, \dirtvar_4,
  (\dirtcoervar_1 : \dirtvar_1 \le \dirtvar_2),
  (\dirtcoervar_2 : \dirtvar_1 \le \dirtvar_3),
  (\dirtcoervar_3 : \dirtvar_4 \le \ops_1 \cup \dirtvar_3),
  (\dirtcoervar_4 : \dirtvar_4 \le \ops_2)
\]
can be captured with the graph
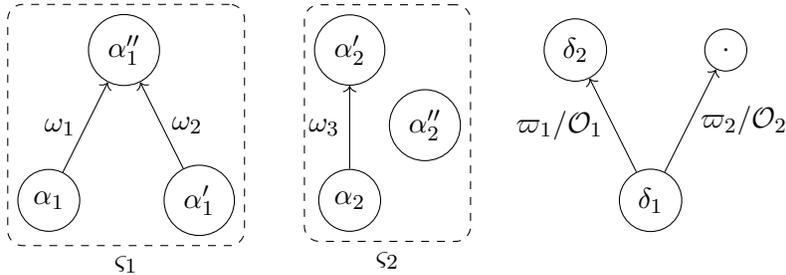
\[
  \begin{tikzpicture}[box/.style = {draw, dashed, rounded corners=2mm}]
    \node [shape=circle, draw] (d1) at (2.5, 0) {$\dirtvar_1$};
    \node [shape=circle, draw] (d2) at (1, 2) {$\dirtvar_2$};
    \node [shape=circle, draw] (d3) at (4, 2) {$\dirtvar_3$};
    \node [shape=circle, draw] (d4) at (5.5, 0) {$\dirtvar_4$};
    \node [shape=circle, draw] (cd) at (7, 2) {$\cdot$};
    \draw[->] (d1) edge node[near end, right] {$\dirtcoervar_1$} (d2);
    \draw[->] (d1) edge node[near end, left] {$\dirtcoervar_2$} (d3);
    \draw[->] (d4) edge node[near start, left] {$\dirtcoervar_3 / \ops_1$} (d3);
    \draw[->] (d4) edge node[near start, right] {$\dirtcoervar_4 / \ops_2$} (cd);
  \end{tikzpicture}
\]
Reduced type parameters and type coercions involve only parameters, and are as such completely independent from dirt coercions.
As a result, we can represent any reduced context with two graphs: one for type coercions, and one for dirt coercions.

We reduce dirt coercion contexts using a function $\subPhase_{dc}(\ctx_{dc}; \ctx_{dc}', \ctx_d', \sigma)$ like before.
Recall that each dirt is either of the form $\ops$ or $\ops \cup \dirtvar$, so in addition to the terminal case~\eqref{case:dc-terminal}, we need to consider four cases, accounting for each option on both sides of the inequality. Furthermore, each case is split into two further subcases, depending if $\ops_1 \subseteq \ops_2$ or not.
First, we have the inequalities where the right-hand side is of the closed form $\ops_2$:
\begin{align*}
  \tag{1}\label{case:dc-terminal}
  \subPhase_{dc}(\emptyctx; \ctx_{dc}', \ctx_d', \sigma)
  & := \tuple{\ctx_{dc}', \ctx_d', \sigma}
  \\[1.5ex]
  \tag{2.1}\label{case:dc-closed-closed-hold}
  \subPhase_{dc}((\ctx_{dc}, \dirtcoervar: \ops_1 \leq \ops_2\text{ for $\ops_1 \subseteq \ops_2$}); \ctx_{dc}', \ctx_d', \sigma)
  & :=
  \\
  & \tag*{\llap{$\subPhase_{dc}(\ctx_{dc}; \ctx_{dc}', \ctx_d', \set{\dirtcoervar \mapsto \ops_1 \cupDirtBoth \emptyset_{\ops_2 \setminus \ops_1}} \circ \sigma)$\qquad\qquad}}
  \\
  \tag{2.2}
  \subPhase_{dc}((\ctx_{dc}, \dirtcoervar: \ops_1 \leq \ops_2\text{ for $\ops_1 \not\subseteq \ops_2$}); \ctx_{dc}', \ctx_d', \sigma)
  & := \text{fail}
  \\[1.5ex]
  \tag{3.1}\label{case:dc-open-closed-hold}
  \subPhase_{dc}((\ctx_{dc}, \dirtcoervar: \ops_1 \cup \dirtvar_1 \leq \ops_2\text{ for $\ops_1 \subseteq \ops_2$}); \ctx_{dc}', \ctx_d', \sigma)
  & :=
  \\
  & \tag*{\llap{$
      \begin{multlined}[t]
        \subPhase_{dc}(\ctx_{dc}; (\ctx_{dc}', \dirtcoervar': \dirtvar_1 \leq \ops_2), \ctx_d', \set{\dirtcoervar \mapsto \ops_1 \cupDirtBoth \dirtcoervar'} \circ \sigma) \\
        \text{ where $\dirtcoervar'$ fresh}
  \end{multlined}$\qquad\qquad}}
  \\
  \tag{3.2}\label{case:dc-open-closed-fail}
  \subPhase_{dc}((\ctx_{dc}, \dirtcoervar: \ops_1 \cup \dirtvar_1 \leq \ops_2\text{ for $\ops_1 \not\subseteq \ops_2$}); \ctx_{dc}', \ctx_d', \sigma)
  & := \text{fail}
\end{align*}
followed by the inequalities where the right hand side $\ops_2 \cup \dirtvar_2$ features a dirt parameter:
\begin{align*}
  \tag{4.1}\label{case:dc-closed-open-hold}
  \subPhase_{dc}((\ctx_{dc}, \dirtcoervar: \ops_1 \leq \ops_2 \cup \dirtvar_2\text{ for $\ops_1 \subseteq \ops_2$}); \ctx_{dc}', \ctx_d', \sigma)
  & := \subPhase_{dc}(\ctx_{dc}; \ctx_{dc}', \ctx_d', \sigma)
  \\
  \tag{4.2}\label{case:dc-closed-open-fail}
  \subPhase_{dc}((\ctx_{dc}, \dirtcoervar: \ops_1 \leq \ops_2 \cup \dirtvar_2\text{ for $\ops_1 \not\subseteq \ops_2$}); \ctx_{dc}', \ctx_d', \sigma)
  & :=
  \\
  & \tag*{\llap{$
      \begin{multlined}[t]
        \subPhase_{dc}(\sigma'(\ctx_{dc}, \ctx_{dc}'); \emptyctx, (\ctx_d' \setminus \{\dirtvar_2\}, \dirtvar_2'), \sigma' \circ \sigma) \\
        \text{ where $\dirtvar_2'$ fresh, $\sigma' = \set{\dirtvar_2 \mapsto (\ops_1 \setminus \ops_2) \cup \dirtvar_2'}$}
  \end{multlined}$\quad}}
  \\[1.5ex]
  \tag{5.1}\label{case:dc-open-open-hold}
  \subPhase_{dc}((\ctx_{dc}, \dirtcoervar: \ops_1 \cup \dirtvar_1 \leq \ops_2 \cup \dirtvar_2\text{ for $\ops_1 \subseteq \ops_2$}); \ctx_{dc}', \ctx_d', \sigma)
  & :=
  \\
  & \tag*{\llap{$
      \begin{multlined}[t]
        \subPhase_{dc}((\ctx_{dc}; (\ctx_{dc}', \dirtcoervar' : \dirtvar_1 \leq (\ops_2 \setminus \ops_1) \cupDirtBoth \dirtvar_2'); \ctx_d', \set{\dirtcoervar \mapsto \ops_1 \cupDirtBoth \dirtcoervar'} \circ \sigma)
          \\
          \text{where $\dirtcoervar'$ fresh}
    \end{multlined}$\quad}}
    \\
    \tag{5.2}\label{case:dc-open-open-fail}
    \subPhase_{dc}((\ctx_{dc}, \dirtcoervar: \ops_1 \cup \dirtvar_1 \leq \ops_2 \cup \dirtvar_2\text{ for $\ops_1 \not\subseteq \ops_2$}); \ctx_{dc}', \ctx_d', \sigma)
    & :=
    \\
    & \tag*{\llap{$
        \begin{multlined}[t]
          \subPhase_{dc}((\sigma'(\ctx_{dc}, \ctx_{dc}')); (\dirtcoervar' : \dirtvar_1 \leq (\ops_2 \setminus \ops_1) \cupDirtBoth \dirtvar_2'),(\ctx_d' \setminus \{\dirtvar_2\}, \dirtvar_2'), \set{\dirtcoervar \mapsto \ops_1 \cupDirtBoth \dirtcoervar'} \circ \sigma' \circ \sigma) \\
          \text{where $\dirtvar_2', \dirtcoervar'$ fresh, $\sigma' = \set{\dirtvar_2 \mapsto (\ops_1 \setminus \ops_2) \cup \dirtvar_2'}$}
    \end{multlined}$\quad}}
  \end{align*}
  In cases~\eqref{case:dc-closed-closed-hold}--\eqref{case:dc-open-closed-fail}, we have only $\ops_2$ on the right-hand side, so $\ops_1 \subseteq \ops_2$ must hold for an instantiation to exist.
  If it holds, we can factor out $\ops_1$ from both sides of the coercion $\dirtcoervar$, however note that for flexibility we keep the less restrictive coercion $\dirtcoervar' : \dirtvar \le \ops_2$ in $\ctx_{dc}'$ rather than $\dirtvar \le \ops_2 \setminus \ops_1$ in case~\eqref{case:dc-open-closed-hold}.

  Cases~\eqref{case:dc-closed-open-hold} and~\eqref{case:dc-open-open-hold} are similar to the ones before.
  In cases~\eqref{case:dc-closed-open-fail} and~\eqref{case:dc-open-open-fail}, however, the dirt parameter $\dirtvar_2$ can absorb the difference $\ops_1 \setminus \ops_2$.
  Thus, we do not fail, but replace $\dirtvar_2$ with $(\ops_1 \setminus \ops_2) \cup \dirtvar_2'$.
  As $\dirtvar_2$ could have appeared in already reduced coercions $\ctx_{dc}'$, we need to reintroduce them back in the unreduced context $\ctx_{dc}$ and start afresh.
  Even though in this case, the reducing context $\ctx_{dc}$ increases in size, the set of operations that parameters in $\ctx_d$ can absorb in the future gets smaller, ensuring termination.

  Finally, we can merge all three stages into a single substitution:
  \begin{align*}
    & \phase_\mathtt{red}(\ctx_s, \ctx_d, \ctx_t, \ctx_{dc}, \ctx_{tc}) =                                                                \\
    & \quad \text{let} \tuple{\ctx_t', \ctx_d', \sigma_t} = \subPhase_t(\ctx_t; \emptyctx, \ctx_d, \emptysub)                            \\
    & \quad \text{let} \tuple{\ctx_{tc}', \ctx_{dc}', \sigma_{tc}} = \subPhase_{tc}(\sigma_t(\ctx_{tc}); \emptyctx, \ctx_{dc}, \sigma_t) \\
    & \quad \text{let} \tuple{\ctx_{dc}'', \ctx_d'', \sigma_{dc}} = \subPhase_{dc}(\ctx_{dc}'; \emptyctx, \ctx_d', \sigma_{tc})          \\
    & \quad \tuple{(\ctx_s, \ctx_d'', \ctx_t', \ctx_{dc}'', \ctx_{tc}'), \sigma_{dc}}
  \end{align*}
  Since $\sigma_t$ maps only type parameters, we need to apply it only to $\ctx_{tc}$, while dirt constraints $\ctx_{dc}$ can be left as is.
  Next, type coercion parameters do not appear anywhere else, so there is no need to apply $\sigma_{tc}$, only pass it to the next stage.
  Finally, even though $\sigma_{dc}$ maps dirt parameters, we can assume that $\ctx_t'$ and $\ctx_{tc}'$ are in canonical form, so there is no need to substitute them.

  \begin{prop}
    If we have $\phase_\mathtt{red}(\ctx) = (\ctx', \sigma)$, then $\validSub{\sigma}{\ctx}{\ctx'}$.
  \end{prop}
  \begin{proof}
    An unexciting step-by-step analysis of each of the three stages~\cite[Section~6.1]{Koprivec_2024}.
  \end{proof}

\section{Polarity}
\label{sec:polarity}
As shown in Section~\ref{sec:overview-simplification}, constraint simplification will not necessarily replace types with equivalent, but with compatible ones.
Instead of checking the compatibility of whole types, it is enough to check the compatibility of their free parameters.
Whether a parameter can be safely increased or decreased depends on its \emph{polarity}, which captures their impact on the whole type.

\subsection{Free parameters}

\begin{defi}
  We define a set of free type and dirt parameters $\FP$ as a pair of sets $\tuple{\FP^+, \FP^-}$ for types, dirts and typing contexts as
  \begin{align*}
    \fp{\tyvar}   & = \tuple{\set{\tyvar}, \emptyset}   & \fp{\tyUnit}   & = \tuple{\emptyset, \emptyset} & \fp{\vty \to \cty}        & = \overline{\fp{\vty}} \cup \fp{\cty} \\
    \fp{\dirtvar} & = \tuple{\set{\dirtvar}, \emptyset} & \fp{\emptyset} & = \tuple{\emptyset, \emptyset} & \fp{\set{\op} \cup \dirt} & = \fp{\dirt}
  \end{align*}
  \begin{align*}
    \fp{\dirty{\vty}{\dirt}} & = \fp{\vty} \cup \fp{\dirt}                                                          \\
    \fp{\emptyctx}           & = \tuple{\emptyset, \emptyset} & \fp{\tyCtx, x: \vty} & = \fp{\tyCtx} \cup \fp{\vty}
  \end{align*}
  where all set operations (union, difference, subset, \dots) on pairs are defined component-wise, while $\overline{\tuple{\FP^+, \FP^-}} = \tuple{\FP^-, \FP^+}$ denotes the swap of the components.
\end{defi}

If a parameter is not present in neither $\FP^+$ nor in $\FP^-$, we call it \emph{neutral}.
If it is present in both $\FP^+$ and in $\FP^-$, we call it \emph{bipolar}.

For example, for
\[
  F = \fp{
    (\tyBool \dirtyTo{\dirtvar} \tyvar)
    \dirtyTo{\dirtvar'}
    \tyvar
  }
\]
we have $F^+ = \set{\tyvar, \dirtvar'}$ and $F^- = \set{\tyvar, \dirtvar}$.

Applying a substitution to a type changes its sets of positive and negative parameters, and the new sets can be computed directly from the original ones as
\[
  \sigma(\FP) = \bigcup_{\tyvar^+ \in \FP^+} \fp{\sigma(\tyvar^+)} \cup \overline{\bigcup_{\tyvar^- \in \FP^-} \fp{\sigma(\tyvar^-)}} \cup \bigcup_{\dirtvar^+ \in \FP^+} \fp{\sigma(\dirtvar^+)} \cup \overline{\bigcup_{\dirtvar^- \in \FP^-} \fp{\sigma(\dirtvar^-)}}
\]

\begin{prop}
  For an arbitrary valid substitution $\validSub{\sigma}{\ctx}{\ctx'}$, the following holds:
  \begin{itemize}
    \item if $\tTyCtx{\ctx}{\tyCtx}$, then $\fp{\sigma(\tyCtx)} = \sigma(\fp{\tyCtx})$,
    \item if $\tVty{\ctx}{\tyCtx}{\vty}$, then $\fp{\sigma(\vty)} = \sigma(\fp{\vty})$,
    \item if $\tVty{\ctx}{\tyCtx}{\cty}$, then $\fp{\sigma(\cty)} = \sigma(\fp{\cty})$.
  \end{itemize}
\end{prop}

\subsection{Coercion families}

As discussed in Section~\ref{sec:overview-simplification}, positive parameters can be safely decreased, and negative parameters can be safely increased.
Consider two instantiations of free parameters of type $\vty = (\tyBool \dirtyTo{\dirtvar} \tyvar) \dirtyTo{\dirtvar'} \tyvar$:
\begin{align*}
  \sigma_1 & = \set{\tyvar \mapsto \tyInt, \dirtvar \mapsto \set{\op, \op'}, \dirtvar' \mapsto \emptyset} \\
  \sigma_2 & = \set{\tyvar \mapsto \tyInt, \dirtvar \mapsto \set{\op}, \dirtvar' \mapsto \set{\op}}
\end{align*}
We then have
\[
  \sigma_1(\vty)
  =
  (\tyBool \longDirtyTo{\set{\op, \op'}} \tyInt) \to \tyInt
  \le
  (\tyBool \longDirtyTo{\set{\op}} \tyInt) \longDirtyTo{\set{\op}} \tyInt
  =
  \sigma_2(\vty)
\]
witnessed by the coercion
\[
  (\refl{\tyBool} \longDirtyTo{\set{\op} \cup \emptyset_{\set{\op'}}} \refl{\tyInt}) \longDirtyTo{\emptyset_{\set{\op}}} \refl{\tyInt}
\]
One can see that the structure of the coercion exactly matches the structure of the parametric type, where at leaves, we have coercions between instantiations of corresponding parameters whose direction is determined by the polarity of parameters.
Note that $\tyvar$ is bipolar, thus it can be assigned only reflexive coercions.

We consolidate all these coercions into a \emph{coercion family} $\coerSubst$, which assigns a type coercion~$\coerSubst(\tyvar)$ to each type parameter $\tyvar$, and a dirt coercion~$\coerSubst(\dirtvar)$ to each dirt parameter $\dirtvar$.
The coercion corresponding to the above example would thus be:
\[
  \coerSubst = \set{\tyvar \mapsto \refl{\tyInt}, \dirtvar \mapsto (\set{\op} \cup \emptyset_{\set{\op'}}), \dirtvar' \mapsto \emptyset_{\set{\op}}}
\]

\begin{defi}
  Take substitutions $\validSub{\sigma_1}{\ctx}{\ctx'}$, $\validSub{\sigma_2}{\ctx}{\ctx'}$ and a set of free parameters $\FP$.
  Then, a family $\coerSubst$ as described above is \emph{a coercion between $\sigma_1$ and $\sigma_2$ relative to $\FP$}, written as:
  \[
    \tCoerSub{\ctx'}{\coerSubst}{\sigma_1 \le_{\FP} \sigma_2}
  \]
  if:
  \begin{itemize}
    \item for all $\tyvar \in \FP^+$, we have $\tCoerVal{\ctx'}{\coerSubst(\tyvar)} \sigma_1(\tyvar) \leq \sigma_2(\tyvar)$
    \item for all $\tyvar \in \FP^-$, we have $\tCoerVal{\ctx'}{\coerSubst(\tyvar)} \sigma_2(\tyvar) \leq \sigma_1(\tyvar)$
    \item for all $\dirtvar \in \FP^+$, we have $\tCoerDirt{\ctx'}{\coerSubst(\dirtvar)} \sigma_1(\dirtvar) \leq \sigma_2(\dirtvar)$
    \item for all $\dirtvar \in \FP^-$, we have $\tCoerDirt{\ctx'}{\coerSubst(\dirtvar)} \sigma_2(\dirtvar) \leq \sigma_1(\dirtvar)$
  \end{itemize}
\end{defi}

Given a coercion family~$\coerSubst$, we can extend dirt coercions~$\coerSubst(\dirt)$ to arbitrary dirts~$\dirt$ as long as $\fp{\dirt} \subseteq \FP$ by:
\begin{align*}
  \coerSubst(\emptyset) & = \refl{\emptyset} & \coerSubst(\set{\op} \cup \dirt) & = \set{\op} \cup \coerSubst(\dirt)
\end{align*}
As expected, coercions on parameters give rise to a coercion on the whole dirt.
\begin{prop}
  For any valid substitutions $\validSub{\sigma_1}{\ctx}{\ctx'}, \validSub{\sigma_2}{\ctx}{\ctx'}$, such that $\tCoerSub{\ctx'}{\coerSubst}{\sigma_1 \le_{\FP} \sigma_2}$,
  and an arbitrary dirt $\tDty{\ctx}{\dirt}$ such that $\fp{\dirt} \subseteq \FP$, we have
  \[
    \tCoerDirt{\ctx'}{\coerSubst(\dirt)}{\sigma_1(\dirt) \leq \sigma_2(\dirt)}
  \]
\end{prop}
Similarly, we can define a value coercion $\coerSubst(\vty)$ and a computation coercion $\coerSubst(\cty)$ for arbitrary $\vty, \cty$ such that $\fp{\vty} \subseteq \FP$ or $\fp{\cty} \subseteq \FP$:

\begin{align*}
  \coerSubst(\tyUnit)       & = \refl{\tyUnit}                        &
  \coerSubst(\vty \to \cty) & = \coerSubst(\vty) \to \coerSubst(\cty) & \coerSubst(\dirty{\vty}{\dirt}) & = \dirty{\coerSubst(\vty)}{\coerSubst(\dirt)} \\
\end{align*}

\begin{prop}
  \label{prop:subtypingOnFreeParams}
  For any valid substitutions $\validSub{\sigma_1}{\ctx}{\ctx'}, \validSub{\sigma_2}{\ctx}{\ctx'}$, such that $\tCoerSub{\ctx'}{\coerSubst}{\sigma_1 \le_{\FP} \sigma_2}$,
  for an arbitrary value type $\tVty{\ctx}{\vty}{\skel}$ such that $\fp{\vty} \subseteq \FP$, and for an arbitrary computation type $\tCty{\ctx}{\cty}{\skel}$ such that $\fp{\cty} \subseteq \FP$, we have
  \[
    \tCoerVal{\ctx'}{\coerSubst(\vty)}{\sigma_1(\vty) \leq \sigma_2(\vty)}
    \qquad
    \text{and}
    \qquad
    \tCoerComp{\ctx'}{\coerSubst(\cty)}{\sigma_1(\cty) \leq \sigma_2(\cty)}
  \]
\end{prop}

In addition, coercion families satisfy a number of meta-theoretical properties that will prove useful when establishing soundness of simplifications.

\begin{prop}
  \label{prop:inverseFreeParams}
  Let $\validSub{\sigma_1}{\ctx}{\ctx'}, \validSub{\sigma_2}{\ctx}{\ctx'}$ be valid substitutions and $\FP$ a set of free parameters such that $\tCoerSub{\ctx'}{\coerSubst}{\sigma_1 \leq_{\FP} \sigma_2}$, then $\tCoerSub{\ctx'}{\coerSubst}{\sigma_2 \leq_{\overline{\FP}} \sigma_1}$.
\end{prop}

\begin{prop}
  \label{prop:substitutionTransitivity}
  Take substitutions $\validSub{\sigma_1}{\ctx}{\ctx'}$, $\validSub{\sigma_2}{\ctx}{\ctx'}$, $\validSub{\sigma_3}{\ctx}{\ctx'}$, and coercion families $\tCoerSub{\ctx'}{\coerSubst_1}{\sigma_1 \leq_{\FP} \sigma_2}$ and $\tCoerSub{\ctx'}{\coerSubst_2}{\sigma_2 \leq_{\FP} \sigma_3}$.
  Then $\tCoerSub{\ctx'}{(\coerSubst_2 \circ \coerSubst_1)}{\sigma_1 \leq_{\FP} \sigma_3}$, where $(\coerSubst_2 \circ \coerSubst_1)(\tyvar) := \coerSubst_2(\tyvar) \circ \coerSubst_1(\tyvar)$ and $(\coerSubst_2 \circ \coerSubst_1)(\dirtvar) := \coerSubst_2(\dirtvar) \circ \coerSubst_1(\dirtvar)$.
\end{prop}

\begin{proof}
  Take an arbitrary positive type parameter $\tyvar \in \FP^+$.
  We immediately get coercions $\tCoerVal{\ctx'}{\coerSubst_1(\tyvar)}{\sigma_1(\tyvar) \leq \sigma_2(\tyvar)}$ and $\tCoerVal{\ctx'}{\coerSubst_2(\tyvar)}{\sigma_2(\tyvar) \leq \sigma_3(\tyvar)}$, thus $\tCoerVal{\ctx'}{\coerSubst_2(\tyvar) \circ \coerSubst_1(\tyvar)}{\sigma_1(\tyvar) \leq \sigma_3(\tyvar)}$ as required.
  The proof for positive dirt parameters, and for negative type and dirt parameters is analogous.
\end{proof}

\begin{prop}
  \label{prop:postcomposition}
  Take substitutions $\validSub{\sigma}{\ctx_1}{\ctx_2}$, $\validSub{\sigma_1}{\ctx_2}{\ctx'}$, $\validSub{\sigma_2}{\ctx_2}{\ctx'}$, and a coercion family $\tCoerSub{\ctx'}{\coerSubst}{\sigma_1 \leq_{\sigma(\FP)} \sigma_2}$, then $\tCoerSub{\ctx'}{\coerSubst \circ \sigma}{(\sigma_1 \circ \sigma) \leq_{\FP} (\sigma_2 \circ \sigma)}$, where $(\coerSubst \circ \sigma)(\tyvar) := \coerSubst(\sigma(\tyvar))$ and $(\coerSubst \circ \sigma)(\dirtvar) := \coerSubst(\sigma(\dirtvar))$.
\end{prop}

\begin{proof}
  Take a positive type parameter $\tyvar \in \FP^+$.
  We have $\fp{\sigma(\tyvar)} \subseteq \sigma(\FP)$ by definition of~$\sigma(\FP)$, thus from Proposition~\ref{prop:subtypingOnFreeParams}, we get $\tCoerVal{\ctx'}{\coerSubst(\sigma(\tyvar))}{\sigma_1(\sigma(\tyvar)) \leq_{\FP} \sigma_2(\sigma(\tyvar))}$ or equivalently $\tCoerVal{\ctx'}{\coerSubst(\sigma(\tyvar))}{(\sigma_1 \circ \sigma)(\tyvar) \leq_{\FP} (\sigma_2 \circ \sigma)(\tyvar)}$.
  The proof for other kinds of parameters is analogous.
\end{proof}

\section{Simplifications}
\label{sec:simplifications}
Finally, we have all the ingredients required for our constraint simplification algorithm.
We split the algorithm into multiple independent phases focusing on specific patterns often present in the inferred constraints.
In Section~\ref{sub:simplificationPhases}, we shall discuss general properties of such phases, and afterwards examine their particular instances.

As we have seen in Section~\ref{sub:reduced-contexts}, arbitrary contexts can be transformed into equivalent ones that are representable as graphs, with parameters as vertices and coercions as edges.
We continue this natural direction in Section~\ref{sub:removing-parallel-edges} by removing loops and parallel edges to obtain simple graphs, and in Section~\ref{sub:contracting-sccs} by contracting strongly connected components to obtain directed acyclic graphs (DAGs).
Finally, we end up with two heuristics that significantly simplify constraint graphs that occur in practice: in Section~\ref{sub:compressing-bridges} we attempt to remove as much of redundant intermediate type constraints as possible, and in Section~\ref{sub:minimizing} we do the same with dirt parameters.

\subsection{Simplification phases}
\label{sub:simplificationPhases}

A \emph{phase}~$\phase$ takes as inputs a parameter context~$\ctx$ and a set of free parameters~$\FP$, denoting the polarity of parameters in the type being simplified, and produces a simplified parameter context~$\ctx'$ together with a substitution $\validSub{\sigma}{\ctx}{\ctx'}$ that can be used to replace the removed parameters.

We shall compare contexts in terms of all their possible \emph{instantiations} in some other context, say $\ctxuse$, which we fix throughout the rest of the paper.
Take a polymorphic value $v$ with parameters in context $\ctx$.
Whenever this value is used in $\ctxuse$, we need to instantiate all its parameters to their counterparts (types, coercions, etc.) that need to be valid in $\ctxuse$.
Instantiating a polymorphic value thus amounts to a substitution $\validSub{\intp}{\ctx}{\ctxuse}$, which we shall label with $\intp$ to highlight its role.

From Proposition~\ref{prop:substitution-composition}, it immediately follows that any instantiation $\validSub{\intp'}{\ctx'}{\ctxuse}$ of the simplified context also produces an instantiation $\validSub{\intp' \circ \sigma}{\ctx}{\ctxuse}$ of the original context.
The converse does not hold: not every instantiation of the original context can be reconstructed from one of the simplified context.
This is for the simple reason that simplifications remove parameters and hence decrease the degrees of freedom.
However, we can recover the lost freedom using subtyping.

\begin{defi}
  \label{def:CompletePhase}
  Phase $\phase$ is \emph{complete} if for any $\phase(\ctx, \FP) = \tuple{\ctx', \sigma}$ and any instantiation $\validSub{\intp}{\ctx}{\ctxuse}$, there exists an instantiation $\validSub{\intp'}{\ctx'}{\ctxuse}$ and a coercion family $\tCoerSub{\ctxuse}{\coerSubst_{\intp}}{(\intp' \circ \sigma) \leq_{\FP} \intp}$.
\end{defi}

\begin{defi}
  Let $\phase_1$ and $\phase_2$ be two phases.
  The \emph{composition of phases} $\phase_2 \circ \phase_1$ is defined as follows:
  \[
    (\phase_2 \circ \phase_1)(\ctx, \FP) = (\ctx'', \sigma_2 \circ \sigma_1) \ \textrm{where } \phase_1(\ctx, \FP) = \tuple{\ctx', \sigma_1} \ \textrm{and } \phase_2(\ctx', \sigma_1(\FP)) = \tuple{\ctx'', \sigma_2}
  \]
\end{defi}

Proposition~\ref{prop:substitution-composition} implies that the composition of two phases is also a phase.

\begin{prop}
  Let $\phase_1$ and $\phase_2$ be two phases. If $\phase_1$ and $\phase_2$ are complete, then $\phase_2 \circ \phase_1$ is also complete.
\end{prop}

\begin{proof}
  Let $\phase_1$ and $\phase_2$ be two complete phases and $\validSub{\intp}{\ctx}{\ctxuse}$ be a valid instantiation.
  We need to show that there exists an instantiation $\validSub{\intp''}{\ctx''}{\ctxuse}$ and a coercion family $\tCoerSub{\ctxuse}{\coerSubst}{(\intp'' \circ \sigma_2 \circ \sigma_1) \leq_{\FP} \intp}$.

  Since $\phase_1$ is complete, there exists $\validSub{\intp'}{\ctx'}{\ctxuse}$ and $\tCoerSub{\ctxuse}{\coerSubst_{\intp}}{(\intp' \circ \sigma_1) \leq_{\FP} \intp}$.
  Furthermore, by completeness of $\phase_2$, there exists an instantiation $\validSub{\intp''}{\ctx''}{\ctxuse}$ and a coercion family $\tCoerSub{\ctxuse}{\coerSubst_{\intp'}}{(\intp'' \circ \sigma_2) \leq_{\sigma_1(\FP)} \intp'}$.
  Proposition~\ref{prop:postcomposition} implies that we have $\tCoerSub{\ctxuse}{\coerSubst_{\intp'} \circ \sigma_1}{(\intp'' \circ \sigma_2 \circ \sigma_1) \leq_{\FP} (\intp' \circ \sigma_1)}$.
  Finally, Proposition~\ref{prop:substitutionTransitivity} implies that, as required, $\tCoerSub{\ctxuse}{(\coerSubst_{\intp'} \circ \sigma_1 \circ \coerSubst_{\intp})}{(\intp'' \circ \sigma_2 \circ \sigma_1) \leq_{\FP} \intp}$.
\end{proof}


\subsection{Removing loops and parallel edges}
\label{sub:removing-parallel-edges}

An observant reader might have noticed that our graphs might have loops if the context contains coercion parameters of the form $\tycoervar : \tyvar \le \tyvar$ or $\dirtcoervar : \dirtvar \le \ops \cup \dirtvar$.
However, these are easily disposed by a strengthening
\begin{align*}
  \sigma_{l} & = \{ \tycoervar \mapsto \refl{\tyvar} \mid (\tycoervar: \tyvar \leq \tyvar) \in \ctx \} \circ \{ \dirtcoervar \mapsto \refl{\dirtvar} \cupDirtSingle \ops \mid (\dirtcoervar: \dirtvar \leq \ops \cup \dirtvar) \in \ctx \}
\end{align*}
which is also the reason behind including $\cupDirtSingle$ amongst our coercion constructors.
Note that the sink node cannot have loops as coercions $\dirtcoervar : \ops_1 \le \ops_2$ are already removed during context reduction.

Furthermore, our graph may have multiple parallel edges.
Let us first consider type coercions.
If we already have $\tycoervar : \tyvar_1 \le \tyvar_2$, an additional parameter $\tycoervar' : \tyvar_1 \le \tyvar_2$ between the same two nodes offers no additional flexibility, so we can simply replace it with $\tycoervar$.
In general, we define a strengthening
\[
  \sigma_t = \{ \tycoervar_{ijk} \mapsto \tycoervar_{ij0} \mid (\tycoervar_{ijk} : \tyvar_i \leq \tyvar_j) \in \ctx \}
\]
where $\tycoervar_{ij0}$ is the chosen representative of all coercions between $\tyvar_i$ and $\tyvar_j$.

For dirt coercions, the situation is more interesting due to the presence of operation sets.
Still, if we have both $\dirtcoervar : \dirtvar_1 \le \ops \cup \dirtvar_2$ and $\dirtcoervar' : \dirtvar_1 \le \ops' \cup \dirtvar_2$, then the upper bound for $\dirtvar_1$ can be refined to $(\ops \cap \ops') \cup \dirtvar_2$, and $\ops \cap \ops'$ can be factored out of both coercions.
Again, for a general context, we define
\[
  \sigma_d = \{ \dirtcoervar_{ijk} \mapsto (\ops_{ijk} \setminus \bigcap_l \ops_{ijl}) \cupDirtSingle \dirtcoervar_{ij} \mid (\dirtcoervar_{ijk} : \dirtvar_i \leq \ops_{ijk} \cup \dirtvar_j) \in \ctx \}
\]
for some fresh $\dirtcoervar_{ij} : \dirtvar_i \le (\bigcap_l \ops_{ijl}) \cup \dirtvar_j$ replacing all coercions between $\dirtvar_i$ and $\dirtvar_j$.
We can apply a similar approach for the sink node and define
\[
  \sigma_s = \{ \dirtcoervar_{ij} \mapsto (\ops_{ij} \setminus \bigcap_k \ops_{ik}) \cupDirtSingle \dirtcoervar_i \mid (\dirtcoervar_{ij} : \dirtvar_i \leq \ops_{ij}) \in \ctx \}
\]
coalescing all edges $\dirtcoervar_{ij} : \dirtvar_i \le \ops_{ij}$ between $\dirtvar_i$ and the sink node into a fresh $\dirtcoervar_i : \dirtvar_i \le \bigcap_k \ops_{ik}$.

Even though $\sigma_d$ and $\sigma_s$ are not strengthenings due to the introduction of additional dirt coercions, it is easy to check that
\[
  \phase_{\mathtt{loopPar}}(\ctx, \_)  = \tuple{(\sigma_s \circ \sigma_d \circ \sigma_t \circ \sigma_l)(\ctx), (\sigma_s \circ \sigma_d \circ \sigma_t \circ \sigma_l)}
\]
is a valid and complete simplification phase.

\subsection{Contracting strongly connected components}
\label{sub:contracting-sccs}

Having obtained a simple directed graph, we can simplify it even further.

Firstly, we can easily see that any strongly connected component (which in practice occurs for recursive functions) in a graph requires all the type parameters to be equal.
In terms of coercions, if $\ctx$ contains type coercions
\[
  \tycoervar_1: \tyvar_1 \leq \tyvar_2, \tycoervar_2: \tyvar_2 \leq \tyvar_3, \dots \tycoervar_n: \tyvar_n \leq \tyvar_1
\]
we can unify all type parameters with a substitution $\{ \tyvar_i \mapsto \tyvar_1 \}_{i = 2}^k$.
If we wish, we can also replace all $\tycoervar_i$ with $\refl{\tyvar_1}$, or we can simply weed them out by re-running $\phase_{\mathtt{loopPar}}$.

The same can be done for cycles in dirt coercions in case no additional operations $\ops$ are present, which is often the case in higher-order functions.

We skip writing out the formal definition of the substitution $\sigma$ as it would offer way more verbosity than insight.
Still, $\sigma$ is a strengthening, and we can again define a complete phase $\phase_{\mathtt{scc}}(\ctx, \_)  = \tuple{\sigma(\ctx), \sigma}$, after which we end with a directed acyclic graph (DAG).

\subsection{Compressing bridges}
\label{sub:compressing-bridges}

Secondly, intermediate type parameters $\tycoervar: \tyvar \leq \tyvar'$ can be merged if no other type parameters are dependent on their specific relationship, in particular if the edge is a \emph{bridge}, i.e. the only incoming edge of $\tyvar'$, or the only outgoing edge of $\tyvar$ (Figure~\ref{fig:typeContractionRefl}).

\begin{figure}[ht!]
  \centering
  \begin{subfigure}[b]{0.45\textwidth}
    \centering
    \begin{tikzpicture}
      \node [shape=circle, draw] (ab) at (2, 2) {$\tyvar$};
      \node [shape=circle, draw] (aa) at (2, 5) {$\tyvar'$};
      \node [right of=aa] {$\notin \FP^-$};
      \draw[->] (ab) edge node[midway, right] {$\tycoervar$} (aa);

      \node [shape=circle, draw] (a1) at (0, 7) {$\tyvar_1''$};
      \node (a2) at (2, 7) {$\cdots$};
      \node [shape=circle, draw] (a3) at (4, 7) {$\tyvar_n''$};
      \draw[->] (aa) edge node[midway, left] {$\tycoervar_1'$} (a1);
      \draw[->] (ab) edge[dashed, bend left=20] node[midway, left] {} (a1);
      \draw[->] (aa) edge node[midway, right] {$\tycoervar_n'$} (a3);
      \draw[->] (ab) edge[dashed, bend right=20] node[midway, right] {} (a3);
      \draw[gray] (aa.south)+(-2mm,-1mm) edge [|->, bend right=20] node [midway, left] {$\sigma$} (ab);

    \end{tikzpicture}
    \caption{Exactly one incoming coercion}
  \end{subfigure}
  \hfill
  \begin{subfigure}[b]{0.45\textwidth}
    \centering
    \begin{tikzpicture}
      \node [shape=circle, draw] (a'1) at (0, 0) {$\tyvar_1''$};
      \node (a'2) at (2, 0) {$\cdots$};
      \node [shape=circle, draw] (a'3) at (4, 0) {$\tyvar_n''$};
      \node [shape=circle, draw] (ab) at (2, 2) {$\tyvar$};
      \node [right of=ab] {$\notin \FP^+$};
      \draw[->] (a'1) edge node[midway, left] {$\tycoervar_1''$} (ab);
      \draw[->] (a'1) edge[dashed, bend left=20] node[midway, left] {} (aa);
      \draw[->] (a'3) edge node[midway, right] {$\tycoervar_n''$} (ab);
      \draw[->] (a'3) edge[dashed, bend right=20] node[midway, right] {} (aa);

      \node [shape=circle, draw] (aa) at (2, 5) {$\tyvar'$};
      \draw[->] (ab) edge node[midway, right] {$\tycoervar$} (aa);
      \draw[gray] (ab.north)+(-2mm,1mm) edge [|->, bend left=20] node [midway, left] {$\sigma$} (aa);
    \end{tikzpicture}
    \caption{Exactly one outgoing coercion}
  \end{subfigure}
  \caption{Two cases for type parameter contraction}
  \label{fig:typeContractionRefl}
\end{figure}

First, consider the case when $\tycoervar: \tyvar \leq \tyvar'$ is the only incoming edge of $\tyvar'$.
In that case, we can safely replace $\tyvar'$ with $\tyvar$ as the only affected constraints are ones of the form $\tyvar' \le \tyvar_i''$, which continue to hold in the new form $\tyvar \le \tyvar_i''$ due to transitivity.
To ensure completeness, we must furthermore ensure that $\tyvar'$ is not negative.

Formally, if $\tyvar'$ has only one incoming edge and $\tyvar' \notin \FP^-$, we can apply the phase
\[
  \phase_{\mathtt{bridgeIn}}(\ctx, \FP) = \tuple{\sigma(\ctx), \sigma} \ \textrm{where}\ \sigma = \set{\tyvar' \mapsto \tyvar}
\]
Conversely, if $\tyvar$ has only outgoing edge and $\tyvar \notin \FP^+$, we can apply the phase
\[
  \phase_{\mathtt{bridgeOut}}(\ctx, \FP) = \tuple{\sigma(\ctx), \sigma} \ \textrm{where}\ \sigma = \set{\tyvar \mapsto \tyvar'}
\]
In both cases, the substitution $\sigma$ is a strengthening, thus we get valid simplification phases.

\begin{prop}
  Simplification phases $\phase_{\mathtt{bridgeIn}}$ and $\phase_{\mathtt{bridgeOut}}$ are complete.
\end{prop}

\begin{proof}
  Since the cases are symmetric, we can consider only $\phase_{\mathtt{bridgeIn}}$.
  Take an arbitrary $\validSub{\intp}{\ctx}{\ctxuse}$ and let us define $\validSub{\intp'}{\sigma(\ctx)}{\ctxuse}$

  Since $\sigma$ affects only type and type coercion parameters, we define $\intp'(\skelvar) = \intp(\skelvar)$ for all skeleton parameters $\skelvar$, and similar for dirt parameters $\dirtvar$ and dirt coercion parameters $\dirtcoervar$.
  For all type parameters $\tyvar'' \ne \tyvar'$ (including $\tyvar$) we similarly define $\intp'(\tyvar'') = \intp(\tyvar'')$, while on $\tyvar'$, we leave $\intp'$ undefined.

  Finally, the only affected type coercion parameters are ones that contained $\tyvar'$.
  By assumption, there is only one where $\tyvar'$ appears on the right-hand side: $\tycoervar : \tyvar \le \tyvar'$, and this maps to $(\tycoervar : \tyvar \le \tyvar) \in \sigma(\ctx)$, which we satisfy by defining $\intp'(\tycoervar) = \refl{\intp'(\tyvar)}$.
  All other coercion parameters have $\tyvar'$ on the left, so are of the form $\tycoervar_i' : \tyvar' \le \tyvar_i''$ for some $\tyvar_i''$, and map to $(\tycoervar_i' : \tyvar \le \tyvar_i'') \in \sigma(\ctx)$.
  For those, we define $\intp'(\tycoervar_i') = \intp(\tycoervar_i') \circ \intp(\tycoervar)$.
  As we had $\tCoerVal{}{\intp(\tycoervar)}{\intp(\tyvar) \le \intp(\tyvar')}$ and $\tCoerVal{}{\intp(\tycoervar_i')}{\intp(\tyvar') \le \intp(\tyvar_i'')}$, and as $\intp'(\tyvar) = \intp(\tyvar)$ and $\intp'(\tyvar_i'') = \intp(\tyvar_i'')$, we indeed have $\tCoerVal{}{\intp'(\tycoervar_i')}{\intp'(\tyvar) \le \intp'(\tyvar_i'')}$.

  Having checked that $\validSub{\intp'}{\sigma(\ctx)}{\ctxuse}$ let us find a coercion family $\coerSubst$ such that $\tCoerSub{\ctxuse}{\coerSubst}{(\intp' \circ \sigma) \leq_{\FP} \intp}$.
  The only parameter we need to check is $\tyvar'$, as on all others, instantiations $\intp' \circ \sigma$ and $\intp$ are the same.
  By assumption, we have $\tyvar' \not\in \FP^-$, so we only need to consider the case $\tyvar' \in \FP^+$.
  However, $\intp'(\sigma(\tyvar')) = \intp'(\tyvar) = \intp(\tyvar)$, thus $\tCoerVal{}{\intp(\tycoervar)}{\intp(\tyvar) \le \intp(\tyvar')}$ is the required coercion $\tCoerVal{}{\coerSubst(\tyvar')}{\intp'(\sigma(\tyvar')) \le \intp(\tyvar')}$.
\end{proof}

If there are multiple bridges, we can repeat the process, though we must make keep in mind that the polarity of the remaining parameter is the union of both polarities, so the polarity condition has to be re-evaluated after each compression.
As before, the same compression can be done for bridges in dirt coercions as long as the edge has no additional operations $\ops$ on the right-hand side.

\subsection{Minimizing operation sets}
\label{sub:minimizing}

Recall that dirt constraint parameters in a typing context $\ctx$ are of the form $\dirtcoervar : \dirtvar \leq \dirtvar' \cup \ops$ or $\dirtcoervar : \dirtvar \leq \ops$.
We have already seen that we can apply $\phase_{\mathtt{scc}}$, $\phase_{\mathtt{bridgeIn}}$, and $\phase_{\mathtt{bridgeOut}}$ in case the sets of operations $\ops$ are empty, for example with higher-order functions that are effectful without mentioning any particular operation $\op$.

In contrast to types, dirts have the least element $\emptyset$, which is often desirable as it signifies lack of effects, leading to efficient compilation~\cite{DBLP:journals/pacmpl/KarachaliasKPS21}.
Furthermore, any dirt parameter appearing only on the left-hand side of dirt inequalities can be safely set to $\emptyset$, as long as it is not negative.
Since setting $\dirtvar \mapsto \emptyset$ trivializes inequalities of the form $\dirtcoervar : \dirtvar \leq \dirtvar' \cup \ops$ where $\dirtvar'$ appeared on the right-hand side, this can cause a chain-reaction of further removals.

More generally, we can simultaneously remove any set of non-negative dirt parameters $\mathcal{D} = \set{\dirtvar \mid \dirtvar \in \ctx, \dirtvar \not\in \FP^-}$ with no incoming edges, i.e. no coercion parameters $\dirtcoervar : \dirtvar' \le \dirtvar \cup \ops$ with $\dirtvar \not\in \mathcal{D}$.
For such a set $\mathcal{D}$, we define a substitution $\sigma$ by $\sigma(\dirtvar) = \emptyset$ for each $\dirtvar \in \ctx$, and by $\sigma(\dirtcoervar) = \emptyset_{\sigma(\dirt)}$ for each $(\dirtcoervar : \dirtvar \le \dirt) \in \mathcal{D}$ where $\dirtvar \in \mathcal{D}$.
As $\sigma$ is a strengthening, we can define a simplification phase
\[
  \phase_{\mathtt{emptyDirt}}(\ctx, \FP) = \tuple{\sigma(\ctx), \sigma}
\]

\begin{prop}
  $\phase_{\mathtt{emptyDirt}}$ is complete.
\end{prop}

\begin{proof}
  Take an arbitrary $\validSub{\intp}{\ctx}{\ctxuse}$.
  The required $\validSub{\intp'}{\sigma(\ctx)}{\ctxuse}$ is thus
  \[
    \intp' = (\intp \setminus \mathcal{D}) \setminus \set{\dirtcoervar : \dirtvar \le \dirt \mid \dirtvar \in \mathcal{D}}
  \]
  As by assumption, parameters from $\mathcal{D}$ did not appear on right-hand sides of dirt coercions, no other parameters apart from the removed ones have been impacted by the substitution.
  Thus, to obtain $\tCoerSub{\ctxuse}{\coerSubst}{(\intp' \circ \sigma) \leq_{\FP} \intp}$, we only need to check how $\coerSubst$ is defined on $\dirtvar \in \mathcal{D}$.
  By assumption, $\dirtvar \not\in \FP^-$, so we only need to consider the case $\dirtvar \in \FP^+$.
  But since $(\intp' \circ \sigma)(\dirtvar) = \intp'(\sigma(\dirtvar)) = \emptyset$, we can define $\coerSubst(\dirtvar)$ to be the empty coercion $\emptyset_{\intp(\dirtvar)}$.
\end{proof}

Like $\phase_{\mathtt{bridgein}}$ and $\phase_{\mathtt{bridgeOut}}$, the phase $\phase_{\mathtt{dirtEmpty}}$ can be repeated iteratively until there is no set satisfying the requirements.
In practice, this can be done efficiently by traversing the quotient graph of strongly connected components in topological ordering.
Dually, there is a greatest dirt $\Sigma$, and we can apply a dual phase $\phase_{\mathtt{fullDirt}}$ for sets of non-positive dirt parameters with no incoming edges.
However, as discussed in Section~\ref{sec:implementation}, this does not bring significant practical benefits, so we refrain from doing it in the implementation.

\section{Denotational semantics}
\label{sec:denotationalSemantics}
To show that simplification phases preserve the semantics of terms, we turn to denotational semantics.

\subsection{Free monad}

As it is standard in the setting of algebraic effects, we are going to interpret effectful computations with the help of a \emph{free monad}.
For any assignment $\monadsig = \set{\op_i : U_i \to V_i}_i$, that maps an operation $\op_i$ to sets $U_i$ and $V_i$, we can define a monad $T_\monadsig$ that takes a set $X$ to the set $T_\monadsig X$, defined inductively as the smallest set, containing:
\begin{itemize}
  \item $\inval(a)$ for each $a \in X$,
  \item $\inop[\op_i](u; \kappa)$ for each $\op_i : U_i \to V_i \in \monadsig$, each $u \in U_i$, and each $\kappa \in V_i \to T_\monadsig X$
\end{itemize}
The inclusion $\inval$ acts as the unit of the monad, while for a map $f : X \to T_\monadsig Y$, its lifting $f^\dag : T_\monadsig X \to T_\monadsig Y$ is defined recursively as
\begin{align*}
  f^\dag(\inval(a))        & = f(a)                          \\
  f^\dag(\inop(u; \kappa)) & = \inop(u; f^\dag \circ \kappa)
\end{align*}

\subsection{Skeletal semantics}

Since the presented simplifications preserve types only up to subtyping, we first introduce a \emph{skeletal} semantics, which disregards effect information, giving us a unified platform on which to compare such changes.

For a parameter context $\ctx$, we shall take its interpretation $\sem[]{\ctx}$ to be a set of all valid assignments $\xi$ that map parameters to their interpretations.
Due to dependencies between parameters, we shall expose $\xi$ in turn, starting with skeleton parameters $\skelvar \in \ctx$, which are mapped to arbitrary sets $\xi(\skelvar)$.

Then, for any $\xi \in \sem[]{\ctx}$, we can interpret each skeleton $\tSty{\ctx}{\skel}$ with a set $\sem{\skel}$, defined as:
\begin{align*}
  {\sem{\skelvar}} & = \xi(\skelvar) & {\sem{\tyUnit}} & = \set{\star} & {\sem{\skel_1 \to \skel_2}} & = {\sem{\skel_1}} \to T_{\sem[]{\Sigma}}(\sem{\skel_2})
\end{align*}
where for a global signature $\Sigma = \set{\op_i : A_i \to B_i}_i$, we define $\sem[]{\Sigma} = \set{\op_i : \sem[]{\skel_i} \to \sem[]{\skel_i'}}_i$, where $\tVty{}{A_i}{\skel_i}$ and $\tVty{}{B_i}{\skel_i'}$ are monomorphic, thus $\xi$ does not play a role in their interpretation.

Next, for each $\tVty{\ctx}{\vty}{\skel}$, we define its \emph{skeletal semantics} $\ssem{\vty} = \sem{\skel}$ and similarly $\ssem{\cty} = \sem{\skel}$ for $\tCty{\ctx}{\cty}{\skel}$.
Since interpretations of types ignore all effect information, types with matching skeletons have the same interpretation.
For this reason, we also do not define any skeletal interpretations of coercions.

Given an assignment $\xi$, skeletal denotational semantics of well-typed terms are defined as maps from interpretations of typing contexts to interpretations of types:
\begin{align*}
  {\ssem{\tVal{\ctx}{\tyCtx}{v}{\vty}}} & : {\ssem{\tyCtx}} \to {\ssem{\vty}} & {\ssem{\tComp{\ctx}{\tyCtx}{c}{\cty}}} & : {\ssem{\tyCtx}} \to {\ssem{\cty}}
\end{align*}
where interpretation of typing contexts~$\ssem{\tyCtx}$ is defined component wise:
\begin{align*}
  {\ssem{x_1 : \vty_1, \dots, x_n : \vty_n}} & := \ssem{\vty_1} \times \dots \times \ssem{\vty_n}
\end{align*}

Interpretations are given in the standard way~\cite{DBLP:journals/corr/BauerP13}.
For a tuple $(a_1, \dots, a_n) = \env \in{\ssem{\tyCtx}}$ interpretations of values are defined as:
\begin{align*}
  {\ssem{\tVal{\ctx}{\tyCtx}{x_i}{\vty_i}}}(\env)                      & = a_i                                                                                    \\
  {\ssem{\tVal{\ctx}{\tyCtx}{\tmUnit}{\tyUnit}}}(\env)                 & = \star                                                                                  \\
  {\ssem{\tVal{\ctx}{\tyCtx}{\fun{(x: \vty)} c}{\vty \to \cty}}}(\env) & = a' \in{\ssem{\vty}} \mapsto {\ssem{\tComp{\ctx}{\tyCtx, x : \vty}{c}{\cty}}}(\env, a') \\
  {\ssem{\tVal{\ctx}{\tyCtx}{\cast{v}{\coerVal}}{\vty'}}}(\env)        & = {\ssem{\tVal{\ctx}{\tyCtx}{v}{\vty}}}(\env)
\end{align*}
where in the last line, we can safely discard $\tCoerVal{\ctx}{\coerVal}{\vty \le \vty'}$ since $\ssem{\vty} = \ssem{\vty'}$. Similarly, interpretations of computations are defined as:
\begin{align*}
  {\ssem{\tComp{\ctx}{\tyCtx}{\return{v}}{\dirty{\vty}{\dirt}}}}(\env)           & = \inval{({\ssem{\tVal{\ctx}{\tyCtx}{v}{\vty}}}(\env))}                                                                                                                                       \\
  {\ssem{\tComp{\ctx}{\tyCtx}{\operation{v}{y : A_{\op}}{c}}{\cty}}}(\env)
  & =
  \begin{multlined}[t]\inop \big({\ssem{\tComp{\ctx}{\tyCtx}{v}{A_{\op}}}}(\env); \\ b \in {\ssem{B_{\op}}} \mapsto {{\ssem{\tComp{\ctx}{\tyCtx, y : B_{\op}}{c}{\cty}}}{(\env,b)}}\big)
  \end{multlined} \\
  {\ssem{\tComp{\ctx}{\tyCtx}{\doin{x}{c_1}{c_2}}{\dirty{\vty_2}{\dirt}}}}(\env) & =
  \begin{multlined}[t]
    \big(a' \in {\ssem{\vty_1}} \mapsto {{\ssem{\tComp{\ctx}{\tyCtx, x : \vty_1}{c_2}{\dirty{\vty_2}{\dirt}}}}{(\env, a')}}\big)^\dag \\
    ({\ssem{\tComp{\ctx}{\tyCtx}{c_1}{{\dirty{\vty_1}{\dirt}}}}}(\env))
  \end{multlined}
  \\
  {\ssem{\tComp{\ctx}{\tyCtx}{v_1~v_2}{\cty}}}(\env)                             & = ({\ssem{\tVal{\ctx}{\tyCtx}{v_1}{\vty \to \cty}}}(\env))({\ssem{\tComp{\ctx}{\tyCtx}{v_2}{\vty}}}(\env))                                                                                    \\
  {\ssem{\tComp{\ctx}{\tyCtx}{\cast{c}{\coerComp}}{\cty'}}}(\env)                & = {\ssem{\tComp{\ctx}{\tyCtx}{c}{\cty}}}(\env)
\end{align*}
where again in the last line, we again discard the coercion $\tCoerComp{\ctx}{\coerComp}{\cty \le \cty'}$ since $\ssem{\cty} = \ssem{\cty'}$.

\subsection{Effectful semantics}
\label{sub:effectful-semantics}

We refine the skeletal semantics into a more precise \emph{effectful} semantics, which also reflects the operations, tracked by the effect system.

To do so, an assignment $\xi \in \sem[]{\ctx}$ maps each dirt parameter $\dirtvar \in \ctx$ to some set of operations $\xi(\dirtvar)$, all of which must appear in the global signature $\Sigma$.
Then, each dirt $\tDty{\ctx}{\dirt}$ can similarly be interpreted with a set of operations $\sem{\dirt}$, given by:
\begin{align*}
  {\sem{\dirtvar}} & = \xi(\dirtvar) & {\sem{\emptyset}} & = \emptyset & {\sem{\set{\op} \cup \dirt}} & = \set{\op} \cup \sem{\dirt}
\end{align*}

Interpreting dirt allows us to also interpret value and computation types.
For each value type $\tVty{\ctx}{\vty}{\skel}$ and each computation type $\tCty{\ctx}{\cty}{\skel}$ we define sets $\sem{\vty}$ and $\sem{\cty}$ together with injections into the skeletal semantics.
To interpret types, any assignment $\xi \in \sem[]{\ctx}$ must map each $(\tyvar : \skel) \in \ctx$ to some set $\xi(\tyvar)$ together with an injection $\inject{\tyvar} : \xi(\tyvar) \hookrightarrow \sem{\skel}$.
Then, interpretations are extended to arbitrary value and computation types by:
\begin{align*}
  {\sem{\tyUnit}}       & = \set{\star}                                                                                                                                                         \\
  {\sem{\vty~!~\dirt}}  & = T_{\sem[]{\Sigma}|_{\sem{\dirt}}}(\sem{\vty})                                                                                                                       \\
  {\sem{\vty \to \cty}} & = \Big\{(f, f') : \big(\ssem{\vty} \to \ssem{\cty}\big) \times \big(\sem{\vty} \to \sem{\cty}\big) \,\Big\vert\, f \circ \inject{\vty} = \inject{\cty} \circ f'\Big\}
\end{align*}
where $\sem[]{\Sigma}|_{\sem{\dirt}}$ is the restriction of $\sem[]{\Sigma}$ to operations from $\sem{\dirt}$, ensuring that any valid computation can trigger only operations listed in $\dirt$.

The injections $\inject{\vty} : \sem{\vty} \hookrightarrow \ssem{\vty}$ and $\inject{\cty} : \sem{\cty} \hookrightarrow \ssem{\cty}$ are defined recursively by:
\begin{align*}
  \inject{\tyUnit}(\star)               & = \star                                                    \\
  \inject{\vty \to \cty}(f, f')         & = f                                                        \\
  \inject{\vty!\dirt}(\inval(a))        & = \inval(a)                                                \\
  \inject{\vty!\dirt}(\inop(u; \kappa)) & = \inop(\inject{A_i}(u); \inject{\vty!\dirt} \circ \kappa)
\end{align*}
where we crucially use the fact that $B_i$ is monomorphic and first-order, and so $\ssem{B_i} = \sem{B_i}$.

More interestingly, functions are interpreted with pairs of functions, the first acting on skeletal and the second on effectful semantics.
Interpreting values of $\vty \to \cty$ with functions $\sem{\vty} \to \sem{\cty}$ is not sufficient because there is no simple way to extend their domain from $\sem{\vty}$ to $\ssem{\vty}$ in order to include them in $\ssem{\vty} \to \ssem{\cty}$.
Even though $\inject{\vty \to \cty}$ projects on the first component, it is still injective.
Indeed, if $\inject{\vty \to \cty}(f_1, f_1') = \inject{\vty \to \cty}(f_2, f_2')$, we first get $f_1 = f_2$ and furthermore $\inject{\cty} \circ f_1' = \inject{\cty} \circ f_2'$ which implies $f_1' = f_2'$ since $\inject{\cty}$ is injective.

Next, we turn to previously ignored coercions.
As dirt is interpreted with sets of operations, dirt coercions are interpreted as injections between them.
Thus, an assignment $\xi \in \sem[]{\ctx}$ maps each parameter $(\dirtcoervar : \dirt \le \dirt') \in \ctx$ to an injection $\xi(\dirtcoervar) : \sem{\dirt} \hookrightarrow \sem{\dirt'}$.
Then, a dirt coercion $\tCoerDirt{\ctx}{\coerDirt}{\dirt \leq \dirt'}$, can be interpreted with an injection $\sem{\coerDirt} : \sem{\dirt} \hookrightarrow \sem{\dirt'}$, defined recursively in the expected way.

Similarly, we use injections to interpret value and computation type coercions.
For the fifth and final kind of parameters, an assignment $\xi \in \sem[]{\ctx}$ maps type coercion parameters $(\tycoervar : \vty \le \vty') \in \ctx$ to injections $\xi(\tycoervar) : \sem{\vty} \hookrightarrow \sem{\vty'}$.
In addition, we require these injections to commute with injections into the shared skeletal semantics:
\[
  \inject{\vty} = \inject{\vty'} \circ \xi(\tycoervar)
\]
which can be summed up with the following commutative diagram
\[
  \begin{tikzcd}[row sep=3em,column sep=3em]
    {\sem{\vty}} \arrow[dr, "\inject{\vty}", hookrightarrow] \arrow[rr, "\xi(\tycoervar)", hookrightarrow]  & & {\sem{\vty'}} \arrow[dl, "\inject{\vty'}", hookrightarrow] \\
    & {\ssem{\vty}}&
  \end{tikzcd}
\]
Then, we can interpret arbitrary value and computation type coercions with injections
\begin{align*}
  \sem{\tCoerVal{\ctx}{\coerVal}{\vty \leq \vty'}} & : \sem{\vty} \hookrightarrow \sem{\vty'} & \sem{\tCoerComp{\ctx}{\coerComp}{\cty \leq \cty'}} & : \sem{\cty} \hookrightarrow \sem{\cty'}
\end{align*}
that again commute with injections into the skeletal semantics:
\begin{align*}
  \inject{\vty} & = \inject{A'} \circ \sem{\coerVal} & \inject{\cty} & = \inject{\cty'} \circ \sem{\coerComp}
\end{align*}
The interpretations are defined as follows:
\begin{align*}
  \sem{\trgUnitRefl}(\star)                      & = \star                                                                     \\
  \sem{\coerVal \to \coerComp}((f, f'))          & = (f, \sem{\coerComp} \circ f' \circ \sem{\coerVal})                        \\
  \sem{\dirty{\coerVal}{\coerDirt}}(\inval(a))   & = \inval(\sem{\coerVal}(a))                                                 \\
  \sem{\dirty{\coerVal}{\coerDirt}}(\inop(u; k)) & = \inop[\sem{\coerDirt}(\op)](u; \sem{\dirty{\coerVal}{\coerDirt}} \circ k)
\end{align*}
The only non-trivial requirement is checking that the function coercion is valid.
Take injections $\sem{\coerVal} : \sem{\vty'} \hookrightarrow \sem{\vty}$ and $\sem{\coerComp} : \sem{\cty} \hookrightarrow \sem{\cty'}$, and take an arbitrary $(f, f') \in \sem{\vty \to \cty}$, i.e. $f \circ \inject{\vty} = \inject{\cty} \circ f'$.
Let us first show that $\sem{\coerVal \to \coerComp}(f, f')$ lies in $\sem{\vty' \to \cty'}$, which follows from
\[
  f \circ \inject{\vty'} =
  f \circ \inject{\vty} \circ \sem{\coerVal} =
  \inject{\cty} \circ f' \circ \sem{\coerVal} =
  \inject{\cty'} \circ \sem{\coerComp} \circ f' \circ \sem{\coerVal}
\]
Showing that $\inject{\vty' \to \cty'} \circ \sem{\coerVal \to \coerComp} = \inject{\vty \to \cty}$ is trivial since the injections just project the first component, while $\sem{\coerVal \to \coerComp}$ preserves it.
As coercions commute with injections into skeletons, their interpretation is uniquely determined.

\begin{prop}
  For an arbitrary $\xi \in \sem[]{\ctx}$, the following holds:
  \begin{itemize}
    \item If $\tCoerDirt{\ctx}{\coerDirt}{\dirt_1 \leq \dirt_2}$ and $\tCoerDirt{\ctx}{\coerDirt'}{\dirt_1 \leq \dirt_2}$, then $\sem{\coerDirt} = \sem{\coerDirt'}$.
    \item If $\tCoerVal{\ctx}{\coerVal}{\vty_1 \leq \vty_2}$ and $\tCoerVal{\ctx}{\coerVal'}{\vty_1 \leq \vty_2}$, then $\sem{\coerVal} = \sem{\coerVal'}$.
    \item If $\tCoerComp{\ctx}{\coerComp}{\cty_1 \leq \cty_2}$ and $\tCoerComp{\ctx}{\coerComp'}{\cty_1 \leq \cty_2}$, then $\sem{\coerComp} = \sem{\coerComp'}$.
  \end{itemize}
\end{prop}

Interpretation for typing contexts is again defined component wise.

\begin{align*}
  {\sem{x_1 : \vty_1, \dots, x_n : \vty_n}} & := \sem{\vty_1} \times \dots \times \sem{\vty_n}
\end{align*}
Applying injections for each component also gives us $\inject{\tyCtx} : \sem{\tyCtx} \hookrightarrow \ssem{\tyCtx}$.

\begin{thm}
  \label{th:sqCommute}
  For any $\tVal{\ctx}{\tyCtx}{v}{\vty}$ or $\tComp{\ctx}{\tyCtx}{c}{\cty}$, and for any $\xi \in \sem[]{\ctx}$, there exist unique maps $\sem{\tVal{\ctx}{\tyCtx}{v}{\vty}} : \sem{\tyCtx} \to \sem{\vty}$ or $\sem{\tComp{\ctx}{\tyCtx}{c}{\cty}} : \sem{\tyCtx} \to \sem{\cty}$ such that the following diagrams commute:
  \[
    \begin{tikzcd}
      {\sem{\tyCtx}} \arrow[r, "{\sem{v}}"] \arrow[d, "\inject{\tyCtx}"', hookrightarrow] &
      {\sem{\vty}} \arrow[d, "\inject{\vty}", hookrightarrow] \\
      {\ssem{\tyCtx}} \arrow[r, "{\ssem{v}}"]                   & {\ssem{\vty}}
    \end{tikzcd}
    \qquad\qquad\qquad
    \begin{tikzcd}
      {\sem{\tyCtx}} \arrow[r, "{\sem{c}}"] \arrow[d, "\inject{\tyCtx}"', hookrightarrow] &
      {\sem{\cty}} \arrow[d, "\inject{\cty}", hookrightarrow] \\
      {\ssem{\tyCtx}} \arrow[r, "{\ssem{c}}"]                   & \ssem{\cty}
    \end{tikzcd}
  \]
\end{thm}

\begin{proof}
  Uniqueness follows from the fact that the connecting mappings are injections.
  The existence is proven by a routine induction on the typing derivation.

  For example, take the case for $\tVal{\ctx}{\tyCtx}{\cast{v}{\coerVal}}{\vty'}$, where $\tVal{\ctx}{\tyCtx}{v}{\vty}$ and $\tCoerVal{\ctx}{\coerVal}{\vty \leq \vty'}$.
  Then, we see that defining $\sem{\tVal{\ctx}{\tyCtx}{\cast{v}{\coerVal}}{\vty'}} = \sem{\tCoerVal{\ctx}{\coerVal}{\vty \leq \vty'}} \circ \sem{\tVal{\ctx}{\tyCtx}{v}{\vty}}$
  makes the following diagram commute, where the left square commutes by induction hypothesis, and the right triangle commutes by definition of $\sem{\coerVal}$.
  \[
    \begin{tikzcd}
      {\sem{\tyCtx}} \arrow[r, "{\sem{v}}"] \arrow[d, "\inject{\tyCtx}"', hookrightarrow] \arrow[rr, "{\sem{\cast{v}{\coerVal}}}", bend left=35] &
      {\sem{\vty}} \arrow[dr, "\inject{\vty}", hookrightarrow] \arrow[r, "{\sem{\coerVal}}", hookrightarrow] &
      {\sem{\vty'}} \arrow[d, "\inject{\vty'}", hookrightarrow] \\
      {\ssem{\tyCtx}} \arrow[rr, "{\ssem{\cast{v}{\coerVal}} = \ssem{v}}"]                   && {\ssem{\vty}}
    \end{tikzcd} \qedhere
  \]
\end{proof}

Recall that apart from Theorem~\ref{thm:substitutionPreservesJudgements}, which shows that valid substitutions preserve typing judgements, Theorem~\ref{th:sqCommute} is the only requirement we require from our terms.
Both being very natural conditions gives us strong confidence that results can be applied to a number of language extensions.

\subsection{Preservation}

We now turn to the crucial soundness result: a coercion family between two substitutions gives us injections between interpretations of substituted terms.

\begin{thm}
  \label{th:denotationPreservation}
  Take substitutions $\validSub{\sigma_1}{\ctx}{\ctx'}$ and $\validSub{\sigma_2}{\ctx}{\ctx'}$ together with a coercion family $\tCoerSub{\ctx'}{\coerSubst}{\sigma_1 \leq_\FP \sigma_2}$.
  Then, for any $\tVal{\ctx}{\tyCtx}{v}{\vty}$ such that $\fp{\vty} \cup \overline{\fp{\tyCtx}} \subseteq \FP$, and any $\xi \in \sem[]{\ctx'}$, we have
  \[
    \sem{\tVal{\ctx'}{\sigma_1(\tyCtx)}{\sigma_1(v)}{\sigma_1(\vty)}} = {\sem{\coerSubst(\vty)}} \circ \sem{\tVal{\ctx'}{\sigma_2(\tyCtx)}{\sigma_2(v)}{\sigma_2(\vty)}} \circ {\sem{\coerSubst(\tyCtx)}}
  \]
  Similarly, for any $\tComp{\ctx}{\tyCtx}{c}{\cty}$ such that $\fp{\cty} \cup \overline{\fp{\tyCtx}} \subseteq \FP$, we have
  \[
    \sem{\tComp{\ctx'}{\sigma_1(\tyCtx)}{\sigma_1(c)}{\sigma_1(\cty)}} = {\sem{\coerSubst(\cty)}} \circ \sem{\tComp{\ctx'}{\sigma_2(\tyCtx)}{\sigma_2(c)}{\sigma_2(\cty)}} \circ {\sem{\coerSubst(\tyCtx)}}
  \]
  where $\coerSubst(\tyCtx)$ is defined component-wise as:
  \[
    \coerSubst({x_1 : \vty_1, \dots, x_n : \vty_n}) := \coerSubst(\vty_1) \times \dots \times \coerSubst(\vty_n)
  \]

\end{thm}

\begin{proof}
  Let us consider only the case for values, as the one for computations proceeds analogously.
  The proof is nicely summed up with the commutative diagram:
  \[
    \tikzcdset{scale cd/.style={every label/.append style={scale=#1},
cells={nodes={scale=#1}}}}
\begin{tikzcd}
  & {\sem{\sigma_1(\tyCtx)}} \arrow[rr, "{{\sem{\sigma_1(v)}}}"] \arrow[dd, "\inject{\sigma_1(\tyCtx)}" description, hookrightarrow] \arrow[ddddd, "{{\sem{\coerSubst(\tyCtx)}}}" description, bend right=60,sloped, hookrightarrow] &          & {\sem{\sigma_1(\vty)}} \arrow[dd, "\inject{\sigma_1(\vty)}" description, hookrightarrow]                                                                                                     &     \\
  &                                                                                                                                                              &         &                                                                                                                                            &     \\
  & {\ssem{\sigma_1(\tyCtx)}} \arrow[rr, "{{\ssem{\sigma_1(v)}}}"] \ar[d,-,double equal sign distance]   &  {} \ar[d,-,double equal sign distance] & {\ssem{\sigma_1(\vty)}} \ar[d,-,double equal sign distance]                                                                                                                               &  \\
  & {\ssem{\sigma_2(\tyCtx)}} \arrow[rr, "{{\ssem{\sigma_2(v)}}}"']                                                                                                                           &    {}     & {\ssem{\sigma_2(\vty)}}                                                                                                                                   &  \\
  &                                                                                                                                                              &  &                                                                                                                                            &     \\
  & {\sem{\sigma_2(\tyCtx)}} \arrow[uu, "\inject{\sigma_2(\tyCtx)}" description, hookrightarrow] \arrow[rr, "{{\sem{\sigma_2(v)}}}"]                                                                   &          & {\sem{\sigma_2(\vty)}} \arrow[uu, "\inject{\sigma_2(\vty)}" description, hookrightarrow] \arrow[uuuuu, "{{\sem{\coerSubst(\vty)}}}" description, bend right=60,sloped,hookrightarrow] &
\end{tikzcd}

  \]
  Here, the top and bottom square commute due to the definition of effectful semantics, the middle square commutes because the skeletal semantics is identical, and the side triangles commute because $\coerSubst$ produces well-typed coercions, which commute with injections into skeletons.
  We conclude by noticing that conclusion follows from the fact that $\inject{\sigma_1(A)}$ is an injection and can be cancelled from the left.
\end{proof}

From this, it immediately follows that complete phases preserve the semantics up to a coercion.
In particular, every instantiation of a polymorphic value~$v$, the scenario we are interested in, can be replaced with a suitably coerced instantiation of a simplified value~$\sigma(v)$.

\begin{cor}
  \label{cor:toplevelCorrectness}
  Let $\tVal{\ctx}{\cdot}{v}{\vty}$ be a well-typed closed value, $\phase$ a complete phase such that $\phase(\ctx, \fp{\vty}) = \tuple{\ctx', \sigma}$.
  Then, for any instantiation $\validSub{\intp}{\ctx}{\ctxuse}$, there exists an instantiation $\validSub{\intp'}{\ctx'}{\ctxuse}$ and a coercion $\tCoerVal{\ctxuse}{\coerVal}{\intp'(\sigma(\vty)) \le \intp(\vty)}$ such that
  \[
    {\sem{\tVal{\ctxuse}{\cdot}{\intp(v)}{\intp(\vty)}}} = {\sem{\tVal{\ctxuse}{\cdot}{\cast{\intp'(\sigma(v))}{\coerVal}}{\intp(\vty)}}}
  \]
  for any assignment $\xi$ of parameters in $\ctxuse$.
\end{cor}

\begin{proof}
  Since $\phase$ is complete, we get $ \tCoerSub{\ctxuse}{\coerSubst}{(\intp' \circ \sigma) \leq \intp}$, and $\tCoerVal{\ctxuse}{\coerSubst(\vty)}{\intp'(\sigma(\vty)) \leq \intp(\vty)}$ is the required coercion $\coerVal$.
  From Theorem~\ref{th:denotationPreservation}, we get
  \begin{align*}
    \sem{\tVal{\ctxuse}{\cdot}{\intp(v)}{\intp(\vty)}}
    & = \sem{\coerSubst(\vty)} \circ \sem{\tVal{\ctxuse}{\cdot}{\intp'(\sigma(v))}{\intp'(\sigma(\vty))}} \\
    & = {\sem{\tVal{\ctxuse}{\cdot}{\cast{\intp'(\sigma(v))}{\coerSubst(\vty)}}{\intp(\vty)}}}
  \end{align*}
  with the second equation being exactly the one as in the proof of Theorem~\ref{th:sqCommute}.
\end{proof}

\section{Implementation}
\label{sec:implementation}
\subsection{Simplification pipeline}

We have implemented the simplification phases described in Section~\ref{sec:simplifications} in the existing \eff compiler as a default post-processing step after top-level type inference~\cite{DBLP:journals/jfp/KarachaliasPSVS20}.
Type inference produces explicitly typed terms in a core calculus, similar to one described in Section~\ref{sec:language} in a given parameter context.
Note that we generalize only top-level values~\cite{DBLP:conf/tldi/VytiniotisJS10}.
From the type of those values, we compute the polarities of parameters, apply all simplification phases, and map the resulting substitution to both the type and the term.
Only after the simplification, we generalize the term and store it for further use.
Due to the Corollary~\ref{cor:toplevelCorrectness}, each time the original top-level value is referenced, a simplified value can be used in its place, with both the instantiation and cast automatically inferred by type inference.
In practice, the casts are usually trivial and optimized away by other source-level transformations~\cite{DBLP:journals/pacmpl/KarachaliasKPS21}.

The implementation of phases is a slightly different from the outline in the Section~\ref{sec:simplifications}.
As constraints are represented with simple directed graphs, phases $\phase_\mathtt{red}$ and $\phase_\mathtt{loopPar}$ are applied implicitly.
To make the implementation simpler, we do not construct substitutions, but only collect equality constraints that we pass to the existing unification engine.
Also, since type and type constraint parameters belonging to different skeletons are independent of each other, we perform all the simplifications separately on smaller graphs, significantly reducing their cost.

As mentioned in Section~\ref{sub:minimizing}, we do not apply the phase $\phase_{\mathtt{fullDirt}}$, that would map all dirt parameters with no outgoing edges to the whole signature $\Sigma$.
First, users can define operations at later time, so $\Sigma$ can grow after the simplifications have been applied.
Next, while removal slightly simplifies dirt coercions, those already disappear when compiling to OCaml, so unlike $\phase_{\mathtt{emptyDirt}}$, there is no additional efficiency benefit.
Finally, we have not encountered such cases in practice.

\subsection{Results}
\label{sub:results}

To evaluate the true impact of our simplifications, we repeated the experiments on the benchmark suite previously used to assess the performance of our optimizing compiler~\cite{DBLP:journals/pacmpl/KarachaliasKPS21}.
This time, we compared manually monomorphised code against both unsimplified and simplified polymorphic code.
The newly added simplification phases successfully reproduced the exact same code as the manually monomorphised version, except that now standard lists are used instead of custom monomorphic ones tailored to specific cases.
For example, the $n$-queens benchmark previously used
\begin{efflisting}
type solution = SolutionNil | SolutionCons of (int * int) * solution
type solutions = SolutionsNil | SolutionsCons of solution * solutions
\end{efflisting}
instead of \effcode{(int * int) list list}, as the use of latter would produce polymorphic functions.

The performance results, shown in Table~\ref{tab:benchmarkTimes}, reflect this equivalence.
All benchmarks were using Bechamel micro-benchmarking tool on OCaml branch \verb|4.12.0+domains+effects|.
We repeated the experiments across a range of problem sizes, but the results were consistent, so we report only the largest instance for each benchmark.

\begin{table}
  \begin{tabular}{rcc}
    \toprule
    & unsimplified & simplified \\
    \midrule
    one solution of $n$-queens                       & 24.5         & 0.92       \\
    all solutions of $n$-queens                      & 489.3        & 0.93       \\
    stateful counter                                 & 2.48         & 1.21       \\
    list of generator values                         & 2.31         & 1.01       \\
    stateful sum of generator values\footnotemark[1] & 1.03         & 0.94       \\
    exceptional arithmetic                           & 1.92         & 1.07       \\
    stateful arithmetic\footnotemark[1]              & 0.86         & 0.86       \\
    pure tree traversal                              & 48.85        & 1.21       \\
    reader tree traversal                            & 42.32        & 0.95       \\
    stateful tree traversal                          & 32.96        & 0.99       \\
    \bottomrule
  \end{tabular}
  \caption{Slowdown factors (lower is better) of polymorphic code in comparison to monomorphic one.}
  \label{tab:benchmarkTimes}
\end{table}

As expected, the runtimes of the simplified code are roughly on par with those of monomorphic code, while unsimplified polymorphic code can be one or two orders of magnitude slower, depending on the degree of polymorphism used in the benchmarks.
Since the only actual difference between the monomorphic and simplified polymorphic code lies in the types used, any performance variation can be attributed to measurement error or to the trade-off between using built-in polymorphic lists, which may benefit from compiler and runtime optimisations, and user-defined monomorphic lists, which may gain from reduced overhead due to their fixed types.

To measure the impact of particular simplification phases, we also compiled the \eff standard library, which features multiple recursive higher-order polymorphic functions, mostly on lists.
Without simplifications, the compiled version of the standard library functions featured 447 explicit coercions, all of which were removed after applying the phases.
\footnotetext[1]{%
  Unfortunately, the two marked benchmarks are affected by an unresolved compiler bug that prevents inlining of a handler, resulting in approximately $700\times$ worse performance compared to hand-written OCaml code without handlers.
  This regression is most likely due to the extensive refactoring that occurred during the transition to a polymorphic compiler.
  Previously, the handler was successfully inlined, yielding comparable performance.
  Due to the excessive runtime, coercions contribute little to the overall cost, and measurement error is amplified.
}
Table~\ref{tab:graphsStdlib} shows the information about graph sizes of simplified constraints sets with different simplification settings (recall that phases $\phase_\mathtt{red}$ and $\phase_\mathtt{loopPar}$ are performed implicitly).
The number of edges in type graphs is shown in \textbf{bold} as these amount to additional coercion parameters in the compiled code.
Additionally, we show the number of generated monadic artifacts that have significant impact on the final runtime efficiency.

\begin{table}[ht!]
  \begin{tabular}{ccccc}
    \toprule
    \textbf{simplification level}                 & \textbf{dirt nodes/edges} & \textbf{type nodes/edges} & \ocamlcode{>>=} & \ocamlcode{return} \\ \midrule
    no simplification                             & 632 / 644                 & 435 / \textbf{447}        & 78              & 29                 \\
    only $\phase_{\mathtt{scc}}$ (dirt and types) & 47 / 53                   & 242 / \textbf{193}        & 19              & 24                 \\
    only dirt simplifications                     & 0 / 0                     & 435 / \textbf{447}        & 27              & 24                 \\
    only type simplifications                     & 632 / 644                 & 0 / \textbf{0}            & 77              & 52                 \\
    all simplifications                           & 0 / 0                     & 0 / \textbf{0}            & 17              & 31                 \\ \bottomrule
  \end{tabular}
  \caption{Impact of different simplification phases on graph sizes and generated artifacts when compiling \eff standard library.}
  \label{tab:graphsStdlib}
\end{table}
Similarly, we compiled polymorphic versions of all the above benchmarks in the suite, with the total results shown in Table~\ref{tab:graphsBenchmarks}.
As with the standard library, the number of nodes, edges and monadic artifacts decreases with each simplification phase.
The final result, where all simplifications are applied, is the same as the one obtained by manually monomorphising the benchmarks.

One should notice the marked difference in the sizes of dirt and type constraints between the standard library and benchmarks.
This is because the benchmarks are heavily skewed towards effects and handlers to showcase optimizations when inlining handlers, while the standard library is geared towards general purpose polymorphic functions (like \texttt{map} and \texttt{fold}).

\begin{table}[ht!]
  \begin{tabular}{ccccc}
    \toprule
    \textbf{simplification level}                 & \textbf{dirt nodes/edges} & \textbf{type nodes/edges} & \ocamlcode{>>=} & \ocamlcode{return} \\ \midrule
    no simplification                             & 4351 / 4832               & 89 / \textbf{98}          & 485             & 253                \\
    only $\phase_{\mathtt{scc}}$ (dirt and types) & 6 / 8                     & 29 / \textbf{20}          & 84              & 211                \\
    only dirt simplifications                     & 0 / 0                     & 89 / \textbf{98}          & 87              & 211                \\
    only type simplifications                     & 4351 / 4832               & 0 / \textbf{0}            & 485             & 265                \\
    all simplifications                           & 0 / 0                     & 0 / \textbf{0}            & 84              & 213                \\ \bottomrule
  \end{tabular}
  \caption{Impact of different simplification phases on graph sizes and generated artifacts when compiling polymorphic versions of \eff benchmark suite.}
  \label{tab:graphsBenchmarks}
\end{table}

We can clearly see how progressive simplification phases reduce the size of the graphs down to zero.
Dirt simplification phases do not reduce coercion parameters as those are irrelevant when compiling to OCaml anyway.
However, they do reduce monadic artifacts as they allow purity based optimizations~\cite{DBLP:journals/pacmpl/KarachaliasKPS21} to kick in and reduce the number of generated binds by more than 60\% in favour of more efficient let bindings.
As expected by Corollary~\ref{cor:toplevelCorrectness}, additional casts increase the number of monadic returns as more terminal values are recognized as pure and need to be lifted to be used in computations that remain monadic.
Note that all those returns were present before, but in form of additional arguments to functions with coercion parameters.

\section{Conclusion}
\label{sec:conclusion}
In this work, we have established a framework in which to explore semantics-preserving simplifications of coercions.
Even though we mentioned effect handlers only briefly, they are our main motivation, as the massive overhead of polymorphic explicit coercions restricted programmers to using monomorphic functions.
With the removal of practically all unnecessary coercion parameters, our work makes the overhead negligible and allows efficient compilation of significantly larger and more modular programs, which is likely to reveal further challenges.

\subsection{Related work}

There are of course multiple other approaches to the efficient evaluation of effect handlers.
From OCaml 5~\cite{DBLP:conf/pldi/Sivaramakrishnan21} onwards, effect handlers are available natively through an efficient runtime representation.
As the main purpose for introducing handlers into OCaml was the support for various concurrency scheduling strategies, the implementation is focused on efficiency of single-shot handlers, i.e. ones where a continuation is called at most once.
In this respect, our approach is agnostic and supports efficient compilation of arbitrary handlers.

Effect-aware compilation is also featured in Koka~\cite{DBLP:journals/corr/Leijen14} that, in addition to an efficient runtime, uses a row-based effect system to determine which functions are pure and do not require a CPS translation~\cite{DBLP:conf/popl/Leijen17}.
In order to speed up the search for the appropriate encompassing handler, Koka uses \emph{evidence passing}, where a vector of active handlers is passed around to operations at runtime~\cite{DBLP:journals/pacmpl/XieL21}.
Similar evidence is determined at compile time in Effekt~\cite{DBLP:journals/pacmpl/BrachthauserSO20}, which enforces lexical scoping of handlers~\cite{DBLP:journals/pacmpl/SchusterBO20, DBLP:conf/pldi/SchusterB0O22} at the slight cost of losing full polymorphism.
It would also be interesting if this approach, which features \emph{subregion evidence}~\cite{DBLP:journals/pacmpl/MullerSSOB23}, a form of explicit subeffecting witnesses,  fits into our simplification framework and could benefit from it in any way.

Another promising venue on the topic of algebraic effects is that of \emph{lift coercions}~\cite{DBLP:journals/pacmpl/BiernackiPPS18,DBLP:journals/pacmpl/BiernackiPPS19,DBLP:journals/pacmpl/XieCIL22,DBLP:journals/pacmpl/YoshiokaSI24}, which enable modular code by explicitly marking operations that should not be handled, thereby preventing their unintentional capture by external handlers.
Similar to our subtyping coercions, lift coercions are explicit coercions between different effects, and it would be interesting to see how much they impact the performance and if our approach could simplify them.

Looking at constraint simplification more generally, the existing literature is vast~\cite{DBLP:journals/tcs/FuhM90,DBLP:conf/sas/TrifonovS96,DBLP:conf/fpca/HengleinR95},
and uses similar techniques as we do: canonising constraints into a graph, followed by graph reduction and context-specific simplifications.
It would be interesting to explore whether bridge compression or our framework of substitutions, coercion families, and complete phases offers any new insights.

One general approach worth mentioning is \emph{garbage collection}~\cite{DBLP:conf/icfp/Pottier96,DBLP:journals/iandc/Pottier01}, a canonical simplification algorithm, which performs just as well as, or better than, a collection of heuristics.
Here, one first transitively closes the constraints, and then keeps only those of the form $\tyvar_1^- \le \tyvar_2^+$.
Even though our earlier work~\cite{DBLP:journals/corr/Pretnar13} was based on garbage collection, this fails to transfer over to explicit coercions.
For example, consider the set of type constraints represented by the following graph:
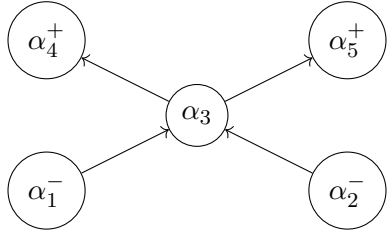
\[
  \begin{tikzpicture}
    \node [shape=circle, draw] (t1) at (0, 0) {$\tyvar_1^-$};
    \node [shape=circle, draw] (t2) at (4, 0) {$\tyvar_2^-$};
    \node [shape=circle, draw] (t3) at (2, 1) {$\tyvar_3$};
    \node [shape=circle, draw] (t4) at (0, 2) {$\tyvar_4^+$};
    \node [shape=circle, draw] (t5) at (4, 2) {$\tyvar_5^+$};
    \draw[->] (t1) -- (t3);
    \draw[->] (t2) -- (t3);
    \draw[->] (t3) -- (t4);
    \draw[->] (t3) -- (t5);
  \end{tikzpicture}
\]
Even if we add additional coercions from $\tyvar_1$ and $\tyvar_2$ to $\tyvar_4$ and $\tyvar_5$ to the graph, the removal of $\tyvar_3$ requires its instantiation, and there is no type we can assign it without changing the meaning.

One could consider adding additional type constructors like $\tyvar_1 \sqcup \tyvar_2$ or $\tyvar_4 \sqcap \tyvar_5$ as done in algebraic subtyping~\cite{DBLP:conf/popl/DolanM17}.
However, recall from Section~\ref{sub:explicit-polymorphism} that in \ocaml, type coercions are translated into polymorphic functions.
Translating $\tyvar_4 \sqcap \tyvar_5$ therefore requires a type constructor \ocamlcode{'a4 /\\ 'a5}, along with coercions \ocamlcode{'a4 /\\ 'a5 -> 'a4} and \ocamlcode{'a4 /\\ 'a5 -> 'a5}, for arbitrary type parameters \ocamlcode{'a4} and \ocamlcode{'a5}.
Type products might serve this purpose, but we have not yet explored this option or the potential performance implications of using (boxed) tuples.

\subsection{Future work}
\label{sub:future-work}

Even though the results in Section~\ref{sub:results} indicate there is not a lot left to simplify, we would like to determine when a given set of constraints cannot be reduced further, like it is done for garbage collection~\cite{DBLP:journals/iandc/Pottier01} or algebraic subtyping~\cite{DBLP:conf/popl/DolanM17}.

Similarly, we would like to explore the role that ordering of simplification phases plays in the end result.
Outside of the framework of simplifying contexts, constraints can be reduced even further through transformations of terms, for example, with suitable rules that rearrange the coercions in terms~\cite{DBLP:journals/pacmpl/KarachaliasKPS21}, changing their local meaning whilst preserving the semantics of the whole program.

Another direction in which polymorphism can be combined with the optimizing compiler~\cite{DBLP:journals/pacmpl/KarachaliasKPS21} is polymorphic function specialisation.
For example, a polymorphic list map of type
\[
  (\tyvar_1 \dirtyTo{\dirtvar} \tyvar_2) \to (\tyvar_1 \, \mathtt{list} \dirtyTo{\dirtvar}\tyvar_2 \, \mathtt{list})
\]
needs to be compiled using the monadic embedding as the given function can be effectful.
Since this is inefficient for pure functions, it is beneficial to compile an additional variant~\cite{DBLP:conf/popl/Leijen17} with the assumption $\dirtvar = \emptyset$.
The number of possible variants is exponential in the number of dirt parameters, so the presented algorithm already helps in this regard.
It remains to be explored, but it is likely that in practice only two variants are sufficient, a general one and one with all dirt parameters set to $\emptyset$.

As the focus of our work was practical --- how to best reduce the number of coercion parameters --- we kept the denotational semantics as simple as necessary to still guarantee soundness of simplifications.
However, the interplay between various constructs hints towards a 2-category with parameter contexts as objects, substitutions as morphisms, and coercion families as 2-morphisms, though the role of polarity remains to be explored.

\section*{Acknowledgments}
\noindent
For their suggestions and comments, we would like to sincerely thank the anonymous reviewers, Danel Ahman, Andrej Bauer, George Karachalias, Sam Lindley, Tom Schrijvers and the participants of IFIP WG~2.1 meeting in Oxford, especially Fritz Henglein.

\bibliographystyle{alphaurl}
\bibliography{references}

\clearpage
\appendix
\section{Well-formedness rules}
\label{sec:wellFormedness}
\begin{figure}[h]
  \begin{mathpar}
    \hilite{\tCtxn{s}{\ctx_s}}\hfill \\ 
    \infer{{}}{\tCtxn{s}{\emptyctx}}
    \and
    \infer{
      \tCtxn{s}{\ctx_s}
    }{
      \tCtxn{s}{\ctx_s, \skelvar}
    }
    \\ 
    \hilite{\tSty{\ctx_s}{\skel}} \ \textrm{assuming} \ \tCtxn{s}{\ctx_s} \hfill \\
    \infer{\skelvar \in \ctx_s}{
      \tSty{\ctx_s}{\skelvar}
    }

    \infer{
    }{\tSty{\ctx_s}{\tyUnit}}

    \infer{\tSty{\ctx_s}{\skel_1} \\ \tSty{\ctx_s}{\skel_2}}{\tSty{\ctx_s}{\skel_1 \to \skel_2}}
    \\
    \hilite{\tCtxn{d}{\ctx_d}}\hfill \\ 
    \infer{{}}{\tCtxn{d}{\emptyctx}}
    \and
    \infer{
      \tCtxn{d}{\ctx_d}
    }{
      \tCtxn{d}{\ctx_d, \dirtvar}
    }
    \\ 
    \hilite{\tDty{\ctx_d}{\dirt}}\ \textrm{assuming} \ \tCtxn{d}{\ctx_d} \hfill \\
    \infer{\dirtvar \in \ctx_d}{
      \tDty{\ctx_d}{\dirtvar}
    }

    \infer{
    }{\tDty{\ctx_d}{\emptyset}}

    \infer{
      (\op : \vty_1 \to \vty_2) \in \Sigma \\ \tDty{\ctx_d}{\dirt}
    }{\tDty{\ctx_d}{\set{\op} \cup \dirt}}
    \\ 
    \hilite{\tCtxn[\ctx_s]{t}{\ctx_t}}\hfill \\
    \infer{}{\tCtxn[\ctx_s]{t}{\emptyctx}}
    \and
    \infer{
      \tCtxn[\ctx_s]{t}{\ctx_t} \\ \tSty{\ctx_s}{\skel}
    }{
      \tCtxn[\ctx_s]{t}{\ctx_t, (\tyvar : \skel)}
    }

  \end{mathpar}
  \caption{Well-formedness rules for contexts, skeletons, and dirts}
  \label{fig:wellformedness-boring}
\end{figure}

\begin{figure}
  \begin{mathpar}
    \hilite{\tCtxn[\ctx_d]{dc}{\ctx_{dc}}}\ \textrm{assuming} \ \tCtxn{d}{\ctx_d}\hfill \\
    \infer{
      \tCtxn[\ctx_d]{dc}{\ctx_{dc}} \\
      \tDty{\ctx_d}{\dirt} \\
      \tDty{\ctx_d}{\dirt'}
    }{
      \tCtxn[\ctx_d]{dc}{\ctx_{dc}, (\dirtcoervar : \dirt \le \dirt')}
    }
    \\ 
    \hilite{\tCtxn[\ctx_s;\ctx_d;\ctx_t]{tc}{\ctx_{tc}}}\ \textrm{assuming} \ \tCtxn{t}{\ctx_t}, \ \tCtxn{d}{\ctx_d} \ \textrm{and} \ \tCtxn[\ctx_s]{d}{\ctx_d}\hfill \\
    \infer{
      \tCtxn[\ctx_s;\ctx_d;\ctx_t]{tc}{\ctx_{tc}} \\
      \tVty{\ctx_s;\ctx_d;\ctx_t}{\vty}{\skel} \\
      \tVty{\ctx_s;\ctx_d;\ctx_t}{\vty'}{\skel}
    }{
      \tCtxn[\ctx_s;\ctx_d;\ctx_t]{tc}{\ctx_{tc}, (\tycoervar : \vty \le \vty')}
    }
  \end{mathpar}
  \caption{Well-formedness rules for coercion contexts}
  \label{fig:wellformedness-coercion-contexts}
\end{figure}

\begin{figure}
  \begin{mathpar}
    \hilite{\tTyCtx{\ctx}{\tyCtx}}\ \textrm{assuming} \ \tCtx{\ctx}\hfill \\

    \infer{{}}{
      \tTyCtx{\ctx}{\emptyctx}
    }

    \infer{\tTyCtx{\ctx}{\tyCtx} \\ \tVty{\ctx}{\vty}{\skel} }{
      \tTyCtx{\ctx}{\tyCtx, (x : \vty)}
    }

  \end{mathpar}
  \caption{Typing context well-formedness}
  \label{tab:contextWellFormedness}
\end{figure}

\end{document}